\documentclass[numbers, 11pt]{article}

\usepackage[margin=0.77in]{geometry}
\pdfoutput=1
\usepackage[hidelinks]{hyperref}
\usepackage{fontenc}
\usepackage{blkarray}
\usepackage{lmodern}
\usepackage{bigstrut}
\usepackage{booktabs}
\usepackage{amsmath}
\usepackage{amsfonts,amssymb,amsthm,mathtools}
\usepackage{xcolor}
\usepackage{graphicx}   
\usepackage{amsbsy} 
\usepackage[numbers,sort&compress]{natbib}

\usepackage{silence}
\usepackage{comment}
\usepackage{algorithm}
\usepackage{tikz}
\usepackage{cancel}
\usepackage{silence}
\usepackage{bbm}

\newtheorem{theorem}{Theorem}[section]
\newtheorem*{theorem*}{Theorem}
\newtheorem{problem}[theorem]{Problem}
\newtheorem{lemma}[theorem]{Lemma}
\newtheorem{fact}[theorem]{Fact}
\newtheorem*{lemma*}{Lemma}
\newtheorem{corollary}[theorem]{Corollary}
\newtheorem{proposition}[theorem]{Proposition}

\newtheorem{remark}[theorem]{Remark}

\newtheorem{definition}[theorem]{Definition}
\newcommand{\tr}{{\rm Tr}} 
\usepackage{caption}
\usepackage{subcaption}
\usepackage{graphicx}
\usepackage{booktabs}
\usepackage{multirow}
\usepackage[T1]{fontenc}
\usepackage[utf8]{inputenc}
\usepackage{physics}
\usepackage{algpseudocode}
\usepackage{geometry}
\date{}
\begin{document}

\title{\vspace{-0.5cm} Optimizing fermionic Hamiltonians with classical interactions}

\author{Maarten E. Stroeks$^{1,2}$, Barbara M. Terhal$^{1,2}$, and Yaroslav Herasymenko$^{1,2,3,4}$
\vspace{0.2cm} \\ 
\small{$^1$Delft Institute of Applied Mathematics, TU Delft, The Netherlands}  ~~\small{$^2$QuTech, TU Delft, The Netherlands} \\
\small{$^3$QuSoft $\&$ CWI, Amsterdam, The Netherlands
}~~\small{$^4$Perimeter Institute, Waterloo, Canada}}

\maketitle

\begin{abstract}
We consider the optimization problem (ground energy search) for fermionic Hamiltonians with classical interactions. This QMA-hard problem is motivated by the Coulomb electron-electron interaction being diagonal in the position basis, a fundamental fact that underpins electronic-structure Hamiltonians in quantum chemistry and condensed matter. We prove that fermionic Gaussian states achieve an approximation ratio of at least 1/3 for such Hamiltonians, independent of sparsity. This shows that classical interactions are sufficient to prevent the vanishing Gaussian approximation ratio observed in SYK-type models.
We also give efficient semi-definite programming algorithms for Gaussian approximations to several families of traceless and positive-semidefinite classically interacting Hamiltonians, with the ability to enforce a fixed particle number. The technical core of our results is the concept of a Gaussian blend, a construction for Gaussian states via mixtures of covariance matrices.

\end{abstract}

\vspace{-1em}
\section{Introduction}\vspace{-0.5em}
In this paper we study energy optimization, or ground energy search, for fermionic Hamiltonians. Mathematically, it means finding the largest\footnote{We flip the sign convention of the Hamiltonian to align with that used in computer science.} eigenvalue of a $2^n$-dimensional Hermitian matrix, which is a low-degree polynomial in fermionic creation and annihilation operators $\{a^\dagger_j, a^{\phantom\dagger}_j\}_{j\in[n]}$:
\begin{align}
   a^\dagger_j a^{\phantom\dagger}_k+a^{\phantom\dagger}_ka^\dagger_j =\delta_{jk}, ~~~a^{\phantom\dagger}_j a^{\phantom\dagger}_k+a^{\phantom\dagger}_k a^{\phantom\dagger}_j=0,~~~n_j\ket{\boldsymbol{x}}= x_j \ket{\boldsymbol{x}},
\end{align} 
where $n_j\equiv a^\dagger_ja_j$ and $\ket{\boldsymbol{x}}$ for $\boldsymbol{x}=(x_1,\ldots,x_n)\in\{0,1\}^n$ are the computational basis states.
Energy optimization is one of the key computational problems in many-body physics and appears in a number of contexts in condensed matter physics and quantum chemistry. 
In general, it is a QMA-hard optimization task \cite{Schuch_2009, QMA-hardness_FH}; in numerical practice, it is being accomplished with a number of approximation tools \cite{szabo1996modern, jensen2017introduction, bartlett2007coupled,  sholl2022density, cirac2021matrix}. 
An interesting goal mathematically is to give rigorous performance guarantees for such approximate methods. 

This text focuses on approximating the highest energy state with a Gaussian (i.e., free-fermionic) state \cite{Terhal_2002, Windt_2021, Surace_2022}. Generally, Gaussian states are defined as the Gibbs states of Hamiltonians which are quadratic in $\{a^\dagger_j,a^{\phantom\dagger}_j\}_{j\in[n]}$; they admit a classically efficient description in terms of a $2n$-sized \textit{covariance matrix} (morally analogous to the stabilizer tableau for stabilizer states).
In the domain of computational many-body physics, the standard method for finding Gaussian ground state approximations is generalized Hartree-Fock, which is heuristic \cite{szabo1996modern, jensen2017introduction}. But in recent years, also \textit{rigorous} guarantees for Gaussian ground state approximations (or lack thereof) have started to appear \cite{BGKT:manybody,Haldar_2021, HO, YSHT:sparse,HPT:improved, basso2024,anschuetz2024, ding2025}. These use a common metric for any optimization method --- \textit{approximation ratio}, i.e., a guarantee on the ratio between the energy of the state obtained by a method, and the true ground energy\footnote{In the fermionic optimization literature, this ratio has been made well-defined by considering traceless Hamiltonians.}.
It was discovered in \cite{Haldar_2021, HO} that for \textit{general} fermionic Hamiltonians, Gaussian states cannot yield ground energy even up to a constant approximation ratio. This dramatic effect – let us call it Gaussian (approximation) breakdown – was demonstrated for the Sachdev-Ye-Kitaev model. It was later extended to some other models which share the feature of having all-to-all, or at least non-sparse, fermion couplings \cite{YSHT:sparse,basso2024, anschuetz2024, ding2025}. This breakdown can be viewed as a heuristic warning sign for optimization of general quantum chemistry Hamiltonians, as those are also strongly interacting and lack sparsity \cite{ding2025,anschuetz2024}. On the other hand, it has been observed that the Hartree-Fock technique yields high approximation ratios in numerical practice \cite{jensen2017introduction, Haldar_2021}. Rigorously speaking, it had not been settled if quantum chemistry Hamiltonians exhibit Gausisan breakdown.

\vspace{0em}
\section{Main results}\vspace{-0.5em}

Our work is motivated by a key difference between the Hamiltonians analyzed in Refs.~\cite{BGKT:manybody,Haldar_2021, HO, YSHT:sparse,basso2024,anschuetz2024, ding2025} and those arising in quantum chemistry. In particular, the quartic terms in real-space discretized electronic-structure Hamiltonians are not generic but \textit{classical}, i.e., diagonal in occupation-number (computational) basis; see Appendix~A in \cite{babbush2018low} and Eq.~(3) in \cite{arguello2019analogue}.
This holds because chemistry interactions physically arise from Coulomb terms, built out of diagonal particle density operators $n(\boldsymbol{r})$. 
Using CIFH as an acronym for `classically interacting fermionic Hamiltonians', we define

\begin{problem}[\textsc{Traceless CIFH Optimization}]
Consider the Hamiltonian 
\begin{equation}
    H = \sum_{(j,k)\in E}w_{j,k}(\mathbbm{1}/4-n_{j}n_{k}) + \sum_{j\in V}\mu_{j}\big(n_{j} - \mathbbm{1}/2\big) + \sum_{(j,k)\in E'}w'_{j,k}\big(- a_{j}^{\dagger}a_{k} - a_{k}^{\dagger}a_{j}\big)
\label{eq:tracelessFMC}
\end{equation} 
with $w_{j,k}\geq 0$, $w_{j,k}'\in \mathbb{R}$, and $\mu_{j} \in \mathbb{R}$, and vertex set $V$ and edge sets $E$ and $E'$. Compute $\lambda_{\max}(H) =\max_{\rho \in \mathbb{C}^{2^n\times 2^n}} \Big\{ \tr \big( \rho H \big) \text{ s.t. } \rho \succeq 0, \: \tr(\rho) = 1 \Big\}$.
\label{prob:maintracelessoptimizationproblem}
\end{problem}
Solving this problem with $1/{\rm poly}(n)$ precision is QMA-hard by adaptation of a result from \cite{QMA-hardness_FH}, see Appendix~\ref{sec:QMA}. Note that $w_{j,k}\geq 0$ implies repulsive interactions as we are looking to maximize the energy and occupying both modes $j$ and $k$ connected by an edge $(j,k)\in E$ \textit{lowers} the energy. In the absence of the quadratic `hopping' terms, $w'_{j,k}=0$, this is a classical QUBO optimization problem \cite{qubo}. As the main result of this work, we show, in Section~\ref{sec:proof_maintraceless}, that
\begin{theorem}
    There exists a pure fermionic Gaussian state $\rho$ that achieves an approximation ratio $\frac{1}{3}$ for \textsc{Traceless CIFH Optimization} (Problem \ref{prob:maintracelessoptimizationproblem}).
    \label{theorem:maintraceless}
\end{theorem}
This implies that quantum chemistry, unlike general fermionic Hamiltonians, does not exhibit a Gaussian breakdown ---even when the Hamiltonian is non-sparse (possibly dense). 

\textit{Proof sketch.}
Split the Hamiltonian of Eq.\,\eqref{eq:tracelessFMC} as $H=H_{\rm quad}+H_{\rm class}$, with the off-diagonal quadratic part $H_{\rm quad}$ and the diagonal (classical) part $H_{\rm class}$,
\begin{align}
    H_{\rm quad}&=\sum_{(j,k)\in E'}w'_{j,k}\big(- a_{j}^{\dagger}a_{k} - a_{k}^{\dagger}a_{j}\big), \label{eq:quadH} \\
    H_{\rm class}&=\sum_{(j,k)\in E}w_{j,k}(\mathbbm{1}/4-n_{j}n_{k}) + \sum_{j\in V}\mu_{j}\big(n_{j} - \mathbbm{1}/2\big).\label{eq:classH}
\end{align}
Each of these Hamiltonians have Gaussian ground states, $\rho_{\rm max}(H_{\rm quad})$ and $\rho_{\rm max}(H_{\rm class})$, as the computational basis states are Gaussian. Therefore, if either $H_{\rm quad}$ or $H_{\rm class}$ is negligible in operator norm, some constant-ratio Gaussian solution can readily be obtained by choosing one of these states. To obtain the stronger Theorem\,\ref{theorem:maintraceless}, which guarantees the constant ratio of $\frac{1}{3}$ regardless of the relative size of $H_{\rm quad}$ and $H_{\rm class}$, a few more steps are needed. One is to observe that $\rho_{\rm max}(H_{\rm class})$ actually vanishes on $H_{\rm quad}$, as we chose it to be off-diagonal. On the flip side, one can modify $\rho_{\rm max}(H_{\rm quad})$, such that $H_{\rm class}$ only contributes to its energy non-negatively. This step is more technical; the key is to modify the covariance matrix of $\rho_{\rm max}(H_{\rm quad})$ such that its first off-diagonal elements are removed. This can be done while preserving the validity ---and Gaussianity--- of the state. Choosing either thus modified solution $\rho_{\rm max}(H_{\rm quad})$, or $\rho_{\rm max}(H_{\rm class})$, allows to guarantee the approximation ratio of $\frac{1}{3}$ at the worst. 

The Gaussian states which \textit{exist} by Theorem~\ref{theorem:maintraceless} should not in general be efficiently constructable. In fact, finding \textit{any} constant-ratio approximation in poly-time (for the classical problem with $H_{\rm quad}=0$) was ruled out under mild complexity-theory assumptions \cite{arora2005non, khot2006sdp}.
But in structured cases this is achievable. Indeed, in Section~\ref{sec:proof_bipartitetraceless} we show 
\begin{theorem}
    There is a deterministic polynomial-time algorithm 
    that outputs a fermionic Gaussian state $\rho$ achieving approximation ratio $\frac{1}{3}$ for \textsc{Traceless CIFH Optimization} (Problem \ref{prob:maintracelessoptimizationproblem}), provided that the graph $G_{\rm class}=(V,(w,E))$ is bipartite. 
    \label{theorem:bipartitetraceless}
\end{theorem}
A simple example of such bipartite interaction graph is a Fermi-Hubbard model with an onsite interaction between spin-up and spin-down electrons, so that the bi-partition is between spin-up and spin-down modes (note that the hopping Hamiltonian remains unconstrained).
Problem~\ref{prob:maintracelessoptimizationproblem} with a bipartite interaction graph stays QMA-hard, and does not need to be sparse.

The key to proving Theorem\,\ref{theorem:bipartitetraceless} is that the global optimum of $H_{\mathrm{class}}$ can be efficiently found, using a linear program which exploits the bipartite structure. This solution can be then used to give a constant-ratio approximating Gaussian, similarly to that of Theorem~\ref{theorem:maintraceless}. We show that this Gaussian is a feasible solution for an $O(n)$-dimensional semi-definite program. This program accounts for $H_{\rm quad}$ in its objective and for $H_{\rm class}$ in its constraints. In practice it yields states with better approximation ratios, and is an interesting subject for future study. We note that if $G_{\rm class}$ had some other structure that allowed the optimum of $H_{\rm class}$ to be efficiently obtained, then Theorem \ref{theorem:bipartitetraceless} would carry over to those cases as well. 

Tracelessness is one of two main conventions which make the approximation ratio well-defined. The other common option is to make every term of the Hamiltonian positive semi-definite, motivating
\begin{problem}[\textsc{Positive Semi-Definite CIFH Optimization}]
Consider the Hamiltonian 
\begin{equation}
    H = \sum_{(j,k)\in E}w_{j,k}\big(\mathbbm{1}-n_{j}n_{k}\big) + \sum_{j}\mu_{j}n_{j} + \sum_{(j,k)\in E'}w'_{j,k}\big(\mathbbm{1} - a_{j}^{\dagger}a_{k} - a_{k}^{\dagger}a_{j}\big)\succeq 0,
\label{eq:tracefulFMC}
\end{equation} 
with $w_{j,k}\geq 0$, $w_{j,k}'\in \mathbb{R}$, and $\mu_{j} \geq 0$ and vertex set $V$, and edge sets $E$ and $E'$. Compute $\lambda_{\max}(H) =\max_{\rho \in \mathbb{C}^{2^n\times 2^n}} \Big\{ \tr \big( \rho H \big) \text{ s.t. } \rho \succeq 0, \: \tr(\rho) = 1 \Big\}$.
\label{prob:maintracefuloptimizationproblem}
\end{problem}
Of course, the tracefulness of this problem does not break the QMA-hardness proved for Problem \ref{prob:maintracelessoptimizationproblem}. Note that in case $\forall (j,k), \; w'_{j,k}=0$ and $\mu_j=0$, Problem \ref{prob:maintracefuloptimizationproblem} is the weighted Max Cut problem. A special case of Problem \ref{prob:maintracefuloptimizationproblem} is \textsc{Fermionic Max Cut}, see Section \ref{sec:fermionicmaxcut}, which, when the graph is a line, coincides with (weighted) \textsc{Quantum Max Cut} \cite{gharibian_parekh}. Unlike in the traceless case, here the goal is guaranteeing not just a constant approximation ratio, but one that is substantially better than that guaranteed by a fully mixed state (in this case $\frac{1}{2}$). 
In Section~\ref{sec:psd}, we ask: can our methods give an interesting Gaussian approximation to this `positive semidefinite' type of optimization? We find
\begin{theorem}
    There is a polynomial-time algorithm that with probability $\Omega(1)$ outputs a Gaussian state $\rho$ that achieves an approximation ratio $0.637$ for \textsc{PSD CIFH Optimization} (Problem \ref{prob:maintracefuloptimizationproblem}).
    \label{theorem:maintraceful}
\end{theorem}
This state can be found using a semi-definite program similar to that implied in Theorem\,\ref{theorem:bipartitetraceless}, and the ratio is guaranteed by a similarly constructed feasible solution. For the classical part of the solution, given a lack of structure, we adapt the Goemans-Williamson approach \cite{GW}. In more structured settings, better algorithms for optimizing $H_{\rm class}$, see e.g. \cite{qubo}, could improve the approximation ratio in Theorem \ref{theorem:maintraceful}. 

The key technical contribution of our work is the concept of a \textit{Gaussian blend}, introduced in Section~\ref{sec:blending}. A Gaussian blend is defined as a Gaussian state whose covariance matrix is a weighted combination of covariance matrices of several input Gaussian states. The modification of the state ${\rho}_{\rm max}(H_{\rm quad})$, hinted at in the proof sketch of Theorem~\ref{theorem:maintraceless}, is in fact given as a Gaussian blend between two states. A more complex Gaussian blend is key to addressing the following problem,
\begin{problem}[q\textsc{-Particle Traceless CIFH Optimization}]
Consider the Hamiltonian 
\begin{equation}
    H = \sum_{(j,k)\in E}w_{j,k}(\mathbbm{1}/4-n_{j}n_{k}) + \sum_{(j,k)\in E'}w'_{j,k}\big(- a_{j}^{\dagger}a_{k} - a_{k}^{\dagger}a_{j}\big)
\label{eq:tracelessFMC_averagehalffilling}
\end{equation} 
with $w_{j,k}\geq 0$, $w_{j,k}'\in \mathbb{R}$, and vertex set $V$, and edge sets $E$ and $E'$. Compute 
\begin{align}
    \lambda_{\max,\langle q\rangle}(H) =\max_{\rho \in \mathbb{C}^{2^n\times 2^n}} \Big\{ \tr \big( \rho H \big) \text{ s.t. } \tr\big(\rho \hat{N} \big) = q,\: \rho \succeq 0, \: \tr(\rho) = 1 \Big\},
\end{align} with $q \in \{0,1,\ldots,\lfloor n/2 \rfloor\}$.
\label{prob:maintracelessoptimizationproblem_averagehalffilling}
\end{problem} 
In this problem, $\hat{N}\equiv\sum_{j\in V}a^\dagger_j a_j$ is the total particle number operator. Essentially, this is Problem~\ref{prob:maintracelessoptimizationproblem} with a constraint that the particle number is equal to $q$ in expectation. This type of an optimization task is inspired by quantum chemistry and condensed matter theory: there the number of fermions is fundamentally a conserved quantity, which is often fixed by the physical setup. 

In Section~\ref{sec:q_particle_proof}, we show
\begin{theorem}
    If $G_{\rm class} = (V,(w,E))$ is bipartite and $q\leq\lfloor n/2\rfloor$, then there exists a fermionic Gaussian state $\rho$ that achieves an approximation ratio $\frac{1}{2\big( (n-2q)/n + 3/2 \big)}$ for Problem \ref{prob:maintracelessoptimizationproblem_averagehalffilling}. Such a state can be obtained in polynomial time. 
    \label{theorem:maintraceless_averagehalffilling}
\end{theorem} 
Due to the constraining nature of the problem, the proof of Theorem~\ref{theorem:maintraceless_averagehalffilling} involves additional technicalities compared to that of Theorem~\ref{theorem:bipartitetraceless}. The semi-definite program which is used to produce the desired state, now includes the $q$-particle condition as a linear constraint; the provided feasible solution is a Gaussian blend involving $\rho_{\rm max}(H_{\rm class})$, $\rho_{\rm max}(H_{\rm quad})$, and a third, auxiliary state. 
Here, the particular condition of $G_{\rm class}$ being bipartite is more essential than in Theorem \ref{theorem:bipartitetraceless}: in addition to being used in the efficient algorithm for the optimization of $H_{\rm class}$, the bipartite structure is used (in a different way) in our proof that the constructed Gaussian state satisfies the $q$-particle constraint.

Please note that the Problem \ref{prob:maintracelessoptimizationproblem_averagehalffilling}, with $\rho$ having $q$ particles only \textit{in expectation}, is a relaxed version of the problem where $\rho$ has \textit{exactly} $q$ particles (is an eigenstate of the number operator). In fact, the latter is the more appropriate constraint in many (albeit not all) physically motivated settings. Finding a constant-ratio Gaussian approximation for this more stringently constrained problem is an interesting open question; our present tools do not allow to give guarantees for such Slater-state-based approximation. It is in principle possible that Slater states still suffer a Gaussian breakdown for the classically interacting fermion optimization (as some weaknesses of Slater state ansatzes compared to general Gaussians were already shown in \cite{BGKT:manybody}).

Besides proving Gaussian approximation ratios, we also give an argument in Appendix \ref{sec:upper-bound} which shows that there are instances of traceless fermionic Hamiltonians with classical interactions where the Gaussian approximation ratio is upper-bounded away from 1 by a constant. Improving such upper-bounding techniques further is an interesting direction for future research.

\section{Preliminaries}
\label{sec:preliminaries}
We consider an $n$-mode fermionic system, which corresponds to a collection of $n$ annihilation operators $a_{j} \in \mathbb{C}^{2^n \times 2^n}$ for $j \in [n]$ (and $n$ Hermitian conjugate creation operators $a_{j}^{\dagger}$). These operators satisfy $\{a_{j},a_{k}^{\dagger}\}:=a_{j}a_{k}^{\dagger} + a_{k}^{\dagger}a_{j} = \delta_{j,k}
\mathbbm{1}$ and $\{a_{j},a_{k}\} = 0$. In addition, there is a vacuum state $\ket{\rm vac}=\ket{{\bf x}=00\ldots 0}$ s.t. $a_{j}\ket{\rm vac} = 0$, $\forall j\in [n]$. The particle number operator is given by $\hat{N} = \sum_{j=1}^{n}a_{j}^{\dagger}a_{j}$.

Equivalently, an $n$-mode fermionic system can be described by $2n$ Majorana operators $c_{j} \in \mathbb{C}^{2^{n}\times 2^{n}}$ with $j\in [2n]$, defined as 
\begin{align}
    c_{2k-1} = a_{k} + a_{k}^{\dagger}, c_{2k} = i(a_{k} - a_{k}^{\dagger}) \label{eq:def-maj}.
\end{align}
These operators are Hermitian and satisfy $\{c_{j},c_{k}\} = 2\delta_{j,k}\mathbbm{1}$ for $j,k \in [2n]$. We note that any transformation $R\in SO(2n)$ of these operators s.t. $\tilde{c}_{j} = \sum_{k}R_{j,k}c_{k}$ preserves these properties and thus gives rise to a new set of Majorana operators $\{\tilde{c}_{j}\}_{j=1}^{2n}$. 

\subsection{Fermionic Gaussian states}
\begin{definition}[Fermionic Gaussian states]
    Given $2n$ Majorana operators $\{c_{j}\}_{j=1}^{2n}$. A fermionic Gaussian state is a (generally mixed) state of the form 
    \begin{equation}
        \rho_{\rm Gauss} \propto \exp\Big( -i\sum_{j\neq k}^{2n}h_{j,k}c_{j}c_{k} \Big),
    \label{eq:fermionicgaussianstate}
    \end{equation}
    where $h$ is a real-valued anti-symmetric matrix.
    \label{def:fermionicgaussianstates}
\end{definition} 
Since $h$ is an anti-symmetric matrix, it can be brought to block-diagonal form by $R\in SO(2n)$
\begin{equation}
    h = R^{T} \: \bigoplus_{j=1}^{n} \begin{pmatrix} 0 & -b_{j} \\ b_{j} & 0 \end{pmatrix} \: R,
    \label{eq:beta_diagonalized}
\end{equation}
with $b_{j}\in \mathbbm{R}$. Therefore, fermionic Gaussian states can be written as 
\begin{equation}
    \rho_{\rm Gauss} = \frac{1}{2^{n}} \prod_{j=1}^{n}\big( \mathbbm{1} + i\lambda_{j}\tilde{c}_{2j-1}\tilde{c}_{2j} \big),\label{eq:fermionicgaussianstate_diagonalized}
\end{equation}
where $\tilde{c}_{j} = \sum_{k}R_{j,k}c_{k}$ and $\lambda_{j} = \tanh(2b_{j}) \in [-1,+1]$. Iff $\rho_{\rm Gauss}$ is a \textit{pure} fermionic Gaussian state, then $\lambda_{j} = \pm 1$ $\forall j\in [n]$ since only then ${\rm Tr}(\rho_{\rm Gauss}^2)=1$.  

\begin{remark}
    Any mixed fermionic Gaussian state is a mixture of pure fermionic Gaussian states. To see this, consider Eq.~\eqref{eq:fermionicgaussianstate_diagonalized}, with for some $j$'s, $-1<\lambda_{j}<+1$ in the decomposition. For each such $j$, one can write $\big( \mathbbm{1} + i\lambda_{j}\tilde{c}_{2j-1}\tilde{c}_{2j} \big) = p_{j}\big( \mathbbm{1} + i\tilde{c}_{2j-1}\tilde{c}_{2j} \big) + (1-p_{j})\big( \mathbbm{1} - i\tilde{c}_{2j-1}\tilde{c}_{2j} \big)$ with $\lambda_{j} = 2p_{j} - 1$, resulting in a mixture over pure fermionic Gaussian states. 
    \label{remark:mixedgaussianismixturepuregaussians}
\end{remark} 

\begin{fact}
    Given a quadratic fermionic Hamiltonian $H = \sum_{j\neq k}^{2n}h_{j,k}\:ic_{j}c_{k}$, with $h$ a real-valued, anti-symmetric matrix. The eigenstates of $H$ are fermionic Gaussian states $\rho_{\rm Gauss} = \frac{1}{2^{n}} \prod_{j=1}^{n}\big( \mathbbm{1} + i\lambda_{j}\tilde{c}_{2j-1}\tilde{c}_{2j} \big)$, with $\lambda_{j} = \pm 1$ $\forall j\in [n]$ and $\tilde{c}_{j} = \sum_{k}R_{j,k}c_{k}$, where $R \in SO(2n)$ block-diagonalizes $h$ as in Eq.~\eqref{eq:beta_diagonalized}. Hence the eigenstates can be obtained in time polynomial in $n$.  
\label{fact:efficientquadclass}
\end{fact}

A particular type of pure fermionic Gaussian state is a Slater determinant state. These states are the eigenstates of particle number conserving free-fermion Hamiltonians, i.e., of Hamiltonians $H$ s.t. $[H,\hat{N}] = 0$. 
\begin{definition}[Slater determinant and classical states]
    A pure Slater determinant state is a state of the form 
    \begin{equation}
\ket{\psi}=\tilde{a}_{1}^{\dagger}\tilde{a}_{2}^{\dagger}\ldots \tilde{a}_{N}^{\dagger}\ket{\rm vac},
    \end{equation}
    where $\tilde{a}_{j} = \sum_{k=1}^{n}U_{j,k}a_{k}$ with $U\in \mathbb{C}^{n\times n}$ a unitary matrix and $N\in [n]$. A particular type of Slater determinant state is a classical state which is of the form 
    \begin{equation}
\ket{\bf x}=a_{j_1}^{\dagger}a_{j_2}^{\dagger}\ldots a_{j_N}^{\dagger}\ket{\rm vac},
    \end{equation}
    where $j_1<j_2<\ldots <j_N \in [n]$, $N\in [n]$ and ${\bf x} = (x_1,\ldots,x_n)$ with $x_{j_1} = \ldots = x_{j_{N}} = 1$ and all other $x_{j}$'s equal to $0$. 
    \label{def:Slaterstates}
\end{definition}
\noindent 
The following fact will be useful later:
\begin{lemma}
    For the density matrix of a Slater determinant state $\rho=\ket{\psi}\bra{\psi}$ one has
    \begin{align}
 \forall j,k\in [n],\quad &i\tr\big( \rho c_{2j-1}c_{2k}  \big) = -i\tr\big( \rho c_{2j}c_{2k-1} \big),         i\tr\big( \rho c_{2j-1}c_{2k-1} \big) = i\tr\big( \rho c_{2j}c_{2k} \big).
    \end{align}
    \label{lemma:Slaterexpectations}
\end{lemma}
\begin{proof}
Since Slater determinant states $\rho$ are eigenstates of the particle number operator $\hat{N}$, $\tr\big(\rho a_{j}a_{k} \big) = \tr\big(\rho a_{j}^{\dagger}a_{k}^{\dagger} \big) = 0$. Using the definition of Majorana operators in Eq.~\eqref{eq:def-maj}, one has
    \begin{align}
        &ic_{2j-1}c_{2k} = - a_{j}a_{k} + a_{j}a_{k}^{\dagger} - a_{j}^{\dagger}a_{k} + a_{j}^{\dagger}a_{k}^{\dagger},  
        \hspace{1.55cm} ic_{2j}c_{2k-1} = - a_{j}a_{k} - a_{j}a_{k}^{\dagger} + a_{j}^{\dagger}a_{k} + a_{j}^{\dagger}a_{k}^{\dagger}, \nonumber \\ 
        &ic_{2j-1}c_{2k-1} = i\big(a_{j}a_{k} + a_{j}a_{k}^{\dagger} + a_{j}^{\dagger}a_{k} + a_{j}^{\dagger}a_{k}^{\dagger}\big),  
        \hspace{1.02cm} ic_{2j}c_{2k} = i\big(-a_{j}a_{k} + a_{j}a_{k}^{\dagger} + a_{j}^{\dagger}a_{k} - a_{j}^{\dagger}a_{k}^{\dagger}\big), 
        \label{eq:Majoranaquads}
    \end{align}
    from which the claim follows.
\end{proof}


\subsection{Covariance matrix}

Any fermionic density matrix $\rho$ can be written as an even polynomial in the Majorana operators $\{c_j\}$, obeying $\rho\succeq 0, \tr(\rho)=1$. 
For any such density matrix we can define a covariance matrix $\Gamma \in \mathbb{R}^{2n \times 2n}$ with entries 
\begin{align}
    \Gamma_{j,k} := \frac{i}{2}\tr \big( \rho [c_{j},c_{k}] \big).
    \label{def:gamma}
\end{align}  
By its definition $\Gamma_{j,k}\in [-1,1],\Gamma_{j,k}=-\Gamma_{k,j}, \Gamma_{k,k}=0$. An anti-symmetric real matrix $\Gamma$ can be block-diagonalized so that
\begin{equation}
    \Gamma = R^{T} \: \bigoplus_{j=1}^{n} \begin{pmatrix} 0 & \lambda_{j} \\ -\lambda_{j} & 0 \end{pmatrix} \: R,   \label{eq:generalblockdiagcovariance}
\end{equation}
with $R\in {\rm SO}(2n)$, $\lambda_{j}\in \mathbbm{R}$. For all $j$, $|\lambda_{j}|\leq 1$ as the expectation values of the rotated Majorana operators $\tilde{c}_j$ still obey $|\langle i \tilde{c}_i \tilde{c}_j \rangle|\leq 1$. Hence for general fermionic states, possibly non-Gaussian, we have $\Gamma^T \Gamma \leq \mathbbm{1}$. On the other hand, one can show that $\rho$ is a fermionic Gaussian state iff $\Gamma\Gamma^{T} = \mathbbm{1}$ \cite{TerhalNoisyFermQC}.  

The following basic result will be used later on. 
\begin{proposition}
    Any matrix $\Gamma \in \mathbb{R}^{2n\times 2n}$ that is anti-symmetric and has no eigenvalues outside of $[-i,+i]$ corresponds to the covariance matrix of a fermionic Gaussian state. \label{prop:gaussianstate_covariancematrix_correspondence}
\end{proposition}
\begin{proof}
    One block-diagonalizes the anti-symmetric matrix $\Gamma$ as in Eq.~\eqref{eq:generalblockdiagcovariance} and the fact that the eigenvalues of $\Gamma$ are within $[-i,+i]$ implies that $\forall j,\; |\lambda_j|\leq 1$. Hence the (unnormalized) fermionic Gaussian state associated with $\Gamma$ in Eq.~\eqref{eq:generalblockdiagcovariance}, is given in Eq.~\eqref{eq:fermionicgaussianstate} 
    with a block-diagonalized $h$ in Eq.~\eqref{eq:beta_diagonalized} with $b_{j} = {\rm arctanh}(\lambda_{j})/2$.    
\end{proof} 

An essential property of fermionic Gaussian states is Wick's theorem, i.e. expectation values of quartic or higher-order correlations can be expressed in terms of entries of the covariance matrix of the fermionic Gaussian state (see e.g. \cite{TerhalNoisyFermQC}). In this work we need this fact only for quartic correlators: 
\begin{proposition}[Wick's Theorem]
    Let $\rho_{\rm Gauss}$ be a fermionic Gaussian state with covariance matrix $\Gamma$. Then
    \begin{equation}
        \tr \big(\rho_{\rm Gauss} \: c_{i}c_{j}c_{k}c_{l}\big) = -\Gamma_{i,j}\Gamma_{k,l} 
        +\Gamma_{i,k}\Gamma_{j,l}
        - \Gamma_{i,l}\Gamma_{j,k},  
        \label{eq:wick}
    \end{equation}
    with $i<j<k<l$. 
    \label{prop:Wickstheorem}
\end{proposition} 

\subsection{Optimization over Gaussian states}
\label{sec:opt-Gauss}
Due to properties of the covariance matrix of a Gaussian fermionic state, the optimization of a general Hermitian (traceless) fermionic Hamiltonian $H$ with quadratic and quartic terms in $\{c_i\}$ over the set of Gaussian fermionic states can be formulated as an optimization of the form \cite{BGKT:manybody}
\begin{align}
    F(\Gamma)=\max_{\Gamma^T \Gamma\leq \mathbbm{1},\Gamma \in \mathcal{L}}\sum_{ijkl} W_{ijkl} \Gamma_{i,j}\Gamma_{k,l}+\sum_{ij}V_{ij} \Gamma_{i,j},
    \label{eq:F}
\end{align}
with real fully anti-symmetric $V_{ij}$ and $W_{ijkl}$.
Here $\mathcal{L}$ is the space of real anti-symmetric $2n \times 2n$ matrices, obeying 
the condition $\Gamma^T \Gamma\leq \mathbbm{1}$. Using the anti-symmetry of $\Gamma$, this is equivalent to $i\Gamma \leq \mathbbm{1}$ as formulated in Proposition \ref{prop:gaussianstate_covariancematrix_correspondence}. It was shown in \cite{BGKT:manybody} that one can rewrite the linear term in $\Gamma_{i,j}$ as part of the quadratic term, we also use the classical version of this trick in the proof of Lemma \ref{lemma:interaction_approx_algorithm}.
Hence, the general problem of optimizing over Gaussian fermionic states is that of a quadratic optimization problem over a convex set of covariance matrices, see also Lemma \ref{lemma:recombinationofgaussianstates}. Due to Remark \ref{remark:mixedgaussianismixturepuregaussians} this optimum is achieved for a pure Gaussian state, one for which $\Gamma^T\Gamma=\mathbbm{1}$, i.e. the eigenvalues of $i\Gamma$ are $\pm 1$. This quadratic optimization problem is generally hard to solve: Ref.~\cite{BGKT:manybody} has considered efficient approximate optimizations via known results in the literature.

\subsection{Optimization of quadratic Hamiltonians as a semi-definite program} 
\label{sec:quadraticoptimization}
One can efficiently optimize quadratic Hamiltonians over Gaussian fermion states which includes additional linear constraints on the covariance matrix of the Gaussian fermionic state. This essentially follows from formulating the optimization as a semi-definite program (SDP), i.e. the quadratic term in $\Gamma$ in Eq.~\eqref{eq:F} is absent. First, we prove

\begin{lemma}
    Any Hermitian matrix $X\in \mathbb{C}^{2n \times 2n}\succeq 0$ with linear constraints $\forall i\neq j\,,X_{i,j}=-X_{j,i}$ and $\forall i, X_{i,i}=1$, can be written as $X=\mathbbm{1}+i\Gamma$, with $\Gamma$ the covariance matrix of a fermionic Gaussian state. \label{lemma:gaussianstate_Amatrix_correspondence}
\end{lemma}
\begin{proof}
    Since $X$ is anti-symmetric on the off-diagonal and equal to $1$ along the diagonal, we have that $X$ w.l.o.g. equals $X=\mathbbm{1} + iB$ with $B$ a real-valued anti-symmetric matrix. Now let us use the following two facts. (1) The eigenvalues of a real-valued anti-symmetric matrix $B$ come in $\pm i\lambda_{j}$ pairs (with $j\in [n]$), with each $\lambda_{j}$ real-valued. (2) $X = \mathbbm{1}+iB \succeq 0$. Therefore, there are no eigenvalues of $B$ outside of the interval $[-i,+i]$. 
    So, $B$ is an anti-symmetric real-valued matrix with no eigenvalues outside of the $[-i,+i]$ interval. 
    Through Proposition \ref{prop:gaussianstate_covariancematrix_correspondence}, $B$ corresponds to a valid covariance matrix $\Gamma$ of a fermionic Gaussian state. 
\end{proof} 

The standard form of a semi-definite program is 
\begin{align*}
\text{maximize} \quad & \mathrm{Tr}(CX) \\
\text{subject to} \quad & \mathrm{Tr}(A_i X) = b_i, \quad i = 1, \dots, m, \\
& X \succeq 0
\end{align*}
where $C,A_i$ and $X$ are Hermitian matrices and $\mathbf{b}\in \mathbb{R}^m$ with $m={\rm poly}(n)$. Clearly, one can choose the set of feasible solutions of a SDP to be of the form $X=\mathbbm{1}+i\Gamma$ by appropriately choosing the equality constraints given by $\{A_i, b_i\}$ to match those in Lemma \ref{lemma:gaussianstate_Amatrix_correspondence} so that $\Gamma$ is the covariance matrix of a fermionic Gaussian state.
Thus we can show the following

\begin{lemma}
    Given a quadratic Hamiltonian on $n$ modes $H = \sum_{j,k}h_{j,k}\: ic_{j}c_{k}$, with real anti-symmetric matrix $h$. The Gaussian state 
    $\rho_{\rm Gauss}$ that maximizes $\tr(\rho_{\rm Gauss} H)$ can be obtained by solving a semi-definite program, hence in time polynomial in $n$, also in the presence of ${\rm poly}(n)$ additional linear constraints on the covariance matrix $\Gamma$ of $\rho_{\rm Gauss}$. 
\label{lemma:constrainedquadraticoptimization}
\end{lemma}

\begin{proof}
   We have $\tr(\rho_{\rm Gauss} H)=\tr(h^T \Gamma)=-\tr(h \Gamma)=\tr(ih X)=\tr(C X)$ where $X=\mathbbm{1}+i\Gamma$ is a feasible solution of the SDP capturing the properties in Lemma \ref{lemma:gaussianstate_Amatrix_correspondence}, and the matrix $C=ih$ is Hermitian. A polynomial number of linear constraints on $\Gamma$ and thus $X$ can be freely added to define the feasible set. 
\end{proof}

\subsection{Blending Gaussian states} 
\label{sec:blending}

Given $m$ Gaussian states $\rho_1,\rho_2,\ldots,\rho_{m}$, their mixture $\sum_{j=1}^{m}p_j\rho_j$ (with $\sum_{j}p_j = 1$) obviously does not need to be a Gaussian state as Wick's theorem in Eq.~(\ref{eq:wick}) does not apply to such mixture. In Ref.~\cite{TerhalNoisyFermQC} such general mixtures were called convex-Gaussian states. Here, we define a Gaussian state obtained by mixing the covariance matrices instead, to which we refer as the {\em blended} Gaussian state. It is straightforward to prove the following as we can cast the (convex) feasible set of a semi-definite program as the set of covariance matrices:

\begin{lemma}
    Given covariance matrices $\Gamma^1,\Gamma^2,$ $\ldots,\Gamma^m$, there exists a fermionic Gaussian state with covariance matrix $\Gamma = \sum_{i=1}^{m}p_{i}\Gamma^i$, for any probability distribution $\{p_{i}\}_{i=1}^{m}$. \label{lemma:recombinationofgaussianstates}
\end{lemma}
\begin{proof}
    Since $\Gamma^1,\ldots,\Gamma^m$ are covariance matrices, $\Gamma = \sum_{i}p_{i}\Gamma^i$ is real-valued and anti-symmetric. Consequently, its eigenvalues come in $\pm i\lambda$ pairs, with $\lambda\in [-1,+1]$. Since $i\Gamma^1,\ldots,i\Gamma^q$ only have eigenvalues in $[-1,+1]$, we have that $\lambda_{\max}(i\Gamma) \leq \sum_{i}p_{i}\lambda_{\max}(i\Gamma^i)\leq \sum_{i}p_{i} = 1$. Therefore, $\Gamma$ has no eigenvalues outside $[-i,+i]$ and is thus a covariance matrix. Through Proposition \ref{prop:gaussianstate_covariancematrix_correspondence}, we can associate a fermionic Gaussian state with $\Gamma$. 
\end{proof}

\section{Proof of Theorem~\ref{theorem:maintraceless}:\\existence of constant-ratio Gaussian approximations.}
\label{sec:proof_maintraceless}

\begin{proof}[Proof of Theorem~\ref{theorem:maintraceless}]
Let us denote the maximum energy classical eigenstate of $H_{\rm class}$ in Eq.~(\ref{eq:classH}) by $\rho_{\rm class}$ and the maximum energy eigenstate of $H_{\rm quad}$ in Eq.~(\ref{eq:quadH}) by $\rho_{\rm quad}$, and their covariance matrices be $\Gamma^{\rm class}$ and $\Gamma^{\rm quad}$ respectively. Both of these states are fermionic Gaussian states. Since $\rho_{\rm class}$ is classical, $\Gamma^{\rm class} = \bigoplus_{j=1}^{n} \Big(\:^{\:\: 0}_{-\lambda_{j}} \:^{\lambda_{j}}_{\:0}\:\Big)$ with $\lambda_{j}\in \{\pm 1\}$. Note that one does not necessarily have an efficient algorithm to compute $\Gamma^{\rm class}$: this will be addressed in Section \ref{sec:interactionoptimization}. We construct a blended Gaussian state with covariance matrix 
\begin{align}
\Gamma = p_{\rm class}\Gamma^{\rm class} + \frac{1-p_{\rm class}}{2}\big( \Gamma^{\rm mediator} + \Gamma^{\rm quad} \big), 
\end{align}
and let $\rho_{\rm Gauss}$ be the associated fermionic Gaussian state. Here, the covariance matrix $\Gamma^{\rm mediator}$ is defined as
\begin{align}
    &\Gamma^{\rm mediator}_{2j-1,2j} = -\Gamma^{\rm quad}_{2j-1,2j}, \: \forall j\in V, \nonumber \\ 
    &\Gamma^{\rm mediator}_{2j,2j-1} = -\Gamma^{\rm quad}_{2j,2j-1}, \: \forall j\in V, \nonumber \\ 
    &\Gamma^{\rm mediator}_{j,k} = 0, \text{ elsewhere}.
\label{eq:mediator}
\end{align}
Since $\bigl\lvert \Gamma^{\rm quad}_{2j-1,2j} \bigr\rvert \leq 1$, the eigenvalues of the anti-symmetric matrix $\Gamma^{\rm mediator}$ lie in $[-i,+i]$ and thus $\Gamma^{\rm mediator}$ is the covariance matrix of some fermionic Gaussian state via Proposition \ref{prop:gaussianstate_covariancematrix_correspondence}. Via Lemma \ref{lemma:recombinationofgaussianstates}, $\Gamma$ is thus a valid covariance matrix. In order to prove what minimum amount of energy it achieves, we consider the following SDP, which is an instance of the SDP discussed in Lemma \ref{lemma:constrainedquadraticoptimization} optimizing $H_{\rm quad}$ in Eq.~\eqref{eq:quadH}. It takes as input the classical optimum $\Gamma^{\rm class}\in \mathbb{R}^{2n\times 2n}$, the edge set $E$ in $H_{\rm class}$ and the parameter $p_{\rm class}\in [0,1]$, and the constraint matrix $C=ih$ corresponds to that of $H_{\rm quad}$ in Eq.~(\ref{eq:quadH}):
\begin{align}
    \max_{X\in \mathbb{C}^{2n \times 2n}} \quad &
    \tr(C X) \nonumber \\
     \text{subject to} \quad
     & X \succeq 0, 
    \nonumber \\ 
   & \forall i, X_{i,i}=1, \forall i \neq j, X_{i,j}=-X_{j,i}, \mbox{(anti-symmetry)} \nonumber \\
     & \mbox{and} \notag \\
    &  X_{2j-1,2j} = i p_{\rm class} \Gamma^{\rm class}_{2j-1,2j}, \text{for all }j\in V, 
    \nonumber \\ 
    &  X_{2j-1,2k} = -X_{2j,2k-1} \text{ for all }(j,k)\in E, \nonumber \\ 
    &  X_{2j-1,2k-1} = X_{2j,2k} \text{ for all }(j,k)\in E. 
    \label{eq:SDPalgorithm}
\end{align}
The covariance matrix $\Gamma$ is constructed to also obey the additional equality constraints in this SDP (and thus corresponds to a feasible solution of the SDP), i.e. for $X=\mathbbm{1}+i\Gamma$ one can argue:
   \begin{enumerate}
   \item Since $\forall j\in V$, $\Gamma^{\rm mediator}_{2j,2j-1} + \Gamma^{\rm quad}_{2j,2j-1}=-\Gamma^{\rm mediator}_{2j,2j-1} - \Gamma^{\rm quad}_{2j,2j-1}=0$ we have $X_{2j-1,2j} = i p_{\rm class}\Gamma^{\rm class}_{2j-1,2j}$.
                 \item Since $\Gamma^{\rm class}$ and $\Gamma^{\rm mediator}$ are zero at all entries other than $(2j-1,2j)$ and $(2j,2j-1)$ $\forall j\in V$, we have $X = i\frac{1-p_{\rm class}}{2}\Gamma^{\rm quad}$ at all off-diagonal entries unequal to $(2j-1,2j)$ and $(2j,2j-1)$ $\forall j\in V$. Since $\rho_{\rm quad}$ is a Slater determinant state, Lemma \ref{lemma:Slaterexpectations} implies $X_{2j-1,2k} = -X_{2j,2k-1}$ and $X_{2j-1,2k-1} = X_{2j,2k}$ for all $(j,k)\in E$.       
    \end{enumerate}
    Since $\Gamma^{\rm class}$ and $\Gamma^{\rm mediator}$ are zero at all entries other than $(2j-1,2j)$ and $(2j,2j-1)$ $j\in V$, $\Gamma$ achieves approximation ratio $(1-p_{\rm class})/2$ on $H_{\rm quad}$, i.e. $\tr (\rho_{\rm Gauss} H_{\rm quad})\geq \frac{1-p_{
    \rm class}}{2}\tr(\rho_{\rm quad}H_{\rm quad})=\frac{1-p_{
    \rm class}}{2}\lambda_{\rm max}(H_{\rm quad})$.

Next, we argue that the expectation on $H_{\rm class}$ in Eq.~\eqref{eq:classH} for {\em any} feasible solution of the SDP in Eq.~\eqref{eq:SDPalgorithm}, and thus also for the optimum of the SDP, can be lower-bounded as follows. Using Wick's theorem, Proposition \ref{prop:Wickstheorem}, any feasible solution achieves expectation
    \begin{align}
        \frac{1}{4}\sum_{(j,k)\in E}w_{j,k}\big(\Gamma_{2j-1,2j}+\Gamma_{2k-1,2k} -\Gamma_{2j-1,2j}\Gamma_{2k-1,2k} - \Gamma_{2j-1,2k}\Gamma_{2j,2k-1} + \Gamma_{2j-1,2k-1}\Gamma_{2j,2k}\big) - \frac{1}{2}\sum_{j\in V}\mu_j \Gamma_{2j-1,2j} \nonumber \\ \geq \frac{1}{4}\sum_{(j,k)\in E}w_{j,k}\big(p_{\rm class}\Gamma^{\rm class}_{2j-1,2j}+p_{\rm class}\Gamma^{\rm class}_{2k-1,2k} -p_{\rm class}^2\Gamma^{\rm class}_{2j-1,2j}\Gamma^{\rm class}_{2k-1,2k} \big) - \frac{1}{2}\sum_{j\in V}\mu_j p_{\rm class}\Gamma^{\rm class}_{2j-1,2j}, 
        \label{eq:intermed}
    \end{align}
    on $H_{\rm class}$, where we have used the conditions in Eq.~\eqref{eq:SDPalgorithm} in the inequality. Note that the final two conditions in Eq.~\eqref{eq:SDPalgorithm} imply $-\Gamma_{2j-1,2k}\Gamma_{2j,2k-1}\geq 0$ and $\Gamma_{2j-1,2k-1}\Gamma_{2j,2k}\geq 0$ for any feasible solution. 
   The final step is to use Eq.~\eqref{eq:intermed} to prove that any feasible solution achieves at least expectation
   \begin{align}
       p_{\rm class}^{2}\lambda_{\max}\big( H_{\rm class} \big), 
   \end{align}
on $H_{\rm class}$, for which we invoke Lemma \ref{lemma:diagonalcostfunctionKbehaviour}, separately proved below with $\Gamma^{\rm class}_{2j,2j-1}=-z_j$. The approximation ratio achieved by optimum of the SDP is thus
    \begin{align}
        \geq \frac{p_{\rm class}^2\lambda_{\max}(H_{\rm class}) + (1-p_{\rm class})/2\:\lambda_{\max}(H_{\rm quad})}{\lambda_{\max}(H)} \geq \frac{p_{\rm class}^{2}\beta + (1-p_{\rm class})/2}{\beta + 1} = f_{\beta}(p_{\rm class}),
        \label{eq:approx-bound}
    \end{align} 
    with $\beta = \lambda_{\max}(H_{\rm class})/\lambda_{\max}(H_{\rm quad})\geq 0$ and $\lambda_{\max}(H) \leq \lambda_{\max}(H_{\rm class}) + \lambda_{\max}(H_{\rm quad})$, and we have used the bounds derived above. $f_{\beta}(p_{\rm class})$ is a convex function of the $p_{\rm class}\in [0,1]$ and hence the optimal value given $\beta$ is achieved at $p_{\rm class}=0$ or $p_{\rm class}=1$. At $\beta=1/2$, $f_{\beta}(p_{\rm class}=0)=f_{\beta}(p_{\rm class}=1)=1/3$ while at other values of $\beta$, $\max_{p_{\rm class}=0,1} f_{\beta}(p_{\rm class})=\max(\frac{\beta}{\beta+1},\frac{1/2}{\beta+1})\geq 1/3$, leading to the lower bound $1/3$ on the Gaussian approximation ratio. The state $\rho_{\rm Gauss}$ is not necessarily a pure fermionic Gaussian state. Through Remark \ref{remark:mixedgaussianismixturepuregaussians}, however, it is a mixture of pure fermionic Gaussian states. Therefore, at least one of the pure fermionic Gaussian states in the mixture achieves approximation ratio at least $1/3$. 
\end{proof}

\begin{lemma}
    Given an interaction graph $G_{\rm class} = ((\mu,V),(w,E))$. Let $F(z_1,\ldots,z_n) = -\frac{1}{4}\sum_{j,k\in E} w_{j,k}(z_j + z_k + z_j z_k)  + \frac{1}{2} \sum_{j\in V}\mu_j z_j$ with $\{z_{j} = \pm 1\}_{j\in V}$. For any assignment $z_{1},\ldots,z_{n}$, we have  $\max\big(F(p\:z_1,\ldots,p\:z_n),$ $F(-p\:z_1,\ldots,-p\:z_n)\big) \geq p^2\max\big(F(z_1,\ldots,z_n),F(-z_1,\ldots,-z_n)\big)$ for $p\in [0,1]$. Clearly, for the optimal assignment $y_1,\ldots,y_n$, we have that $F\big(p\:y_1,\ldots,p\:y_n\big) \geq p^{2}F\big(y_1,\ldots,y_n\big)$. 
    \label{lemma:diagonalcostfunctionKbehaviour}
\end{lemma}
\begin{proof}
    We define 
    \begin{align}
        F_{1}(z_1,\ldots,z_n) :=&\: -\frac{1}{4}\sum_{j,k\in E} w_{j,k}(z_j + z_k) + \frac{1}{2} \sum_{j\in V}\mu_j z_j, \nonumber \\ 
        F_{2}(z_1,\ldots,z_n) :=&\: -\frac{1}{4}\sum_{j,k\in E} w_{j,k} z_j z_k,
    \end{align}
    so that we have $F(z_1,\ldots,z_n) = F_{1}(z_1,\ldots,z_n) + F_{2}(z_1,\ldots,z_n)$, $F_1(-z_1,\ldots,-z_n) = -F_1(z_1,\ldots,z_n)$ and $F_2(-z_1,\ldots,-z_n) = +F_2(z_1,\ldots,z_n)$. 
    Note that $\max \big(F_1(z_1,\ldots,z_n),F_1(-z_1,\ldots,-z_n)\big) \geq 0$, so that $\max \big(F_1(p\:z_1,\ldots,p\:z_n),F_1(-p\:z_1,\ldots,-p\:z_n)\big)\geq$ 
    $p^2\max\big(F_1(z_1,\ldots,z_n),F_1(-z_1,\ldots,-z_n)\big)$ for $p\in [0,1]$. Hence, for any assignment $z_1,\ldots,z_n$ 
    \begin{align}
        \max\big( F(p\: z_1,\ldots,p\: z_n),F(-p\: z_1,\ldots,-p\: z_n) \big) \geq&\: 
        \\ \nonumber p^2\max\big( F_1( z_1,\ldots,z_n),F_1(-z_1,\ldots,-z_n) \big) + p^2F_2(z_1,\ldots,z_n) =&\: \nonumber \\  p^2\max\big(F(z_1,\ldots,z_n),F(-z_1,\ldots,-z_n)\big).
    \end{align}
\end{proof} 

We can prove a small standalone corollary to Theorem~\ref{theorem:maintraceless} on the scaling of the maximum eigenvalue of the Hamiltonians in Problem \ref{prob:maintracelessoptimizationproblem}. Namely, this scaling is fully determined by $\lambda_{\max}(H_{\rm class})$ and $\lambda_{\max}(H_{\rm quad})$, and there is little frustration between the two contributions $H_{\rm class}$ and $H_{\rm quad}$. 
\begin{corollary}
    For $H = H_{\rm class} + H_{\rm quad}$ as in Problem \ref{prob:maintracelessoptimizationproblem} \textsc{Traceless CIFH Optimization}, we can bound $\frac{\lambda_{\max}(H_{\rm class}) + \lambda_{\max}(H_{\rm quad})}{3} \leq \lambda_{\max}(H) \leq \lambda_{\max}(H_{\rm class}) + \lambda_{\max}(H_{\rm quad})$, so that $\lambda_{\max}(H) = \Theta\big( \lambda_{\max}(H_{\rm class}) + \lambda_{\max}(H_{\rm quad}) \big)$. 
\end{corollary}
\begin{proof}
    The upper bound $\lambda_{\max}(H) \leq \lambda_{\max}(H_{\rm class}) + \lambda_{\max}(H_{\rm quad})$ simply follows from the triangle inequality. For the lower bound $\lambda_{\max}(H) \geq 1/3\big( \lambda_{\max}(H_{\rm class}) + \lambda_{\max}(H_{\rm quad}) \big)$, note that the fermionic Gaussian state $\rho_{\rm Gauss}$ in the proof of Theorem~\ref{theorem:maintraceless} actually achieves $\tr\big(\rho_{\rm Gauss}H \big) \geq 1/3\big(\lambda_{\max}(H_{\rm class}) + \lambda_{\max}(H_{\rm quad})\big)$ through Eq.~\eqref{eq:approx-bound}, so that $\lambda_{\max}(H) \geq 1/3\big( \lambda_{\max}(H_{\rm class}) + \lambda_{\max}(H_{\rm quad}) \big)$.   
\end{proof}

\begin{figure}[b!]
  \centering
  \includegraphics[width=0.6\linewidth]{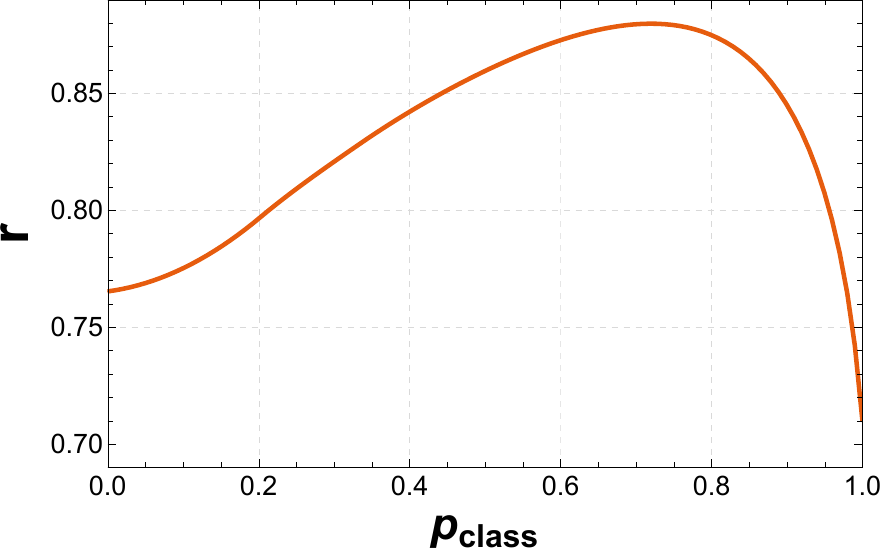}
  \caption{Approximation ratio $r$ versus $p_{\mathrm{class}}$ for $3$-site (i.e., $6$-mode) traceless Fermi-Hubbard Hamiltonian. The maximum approximation ratio $r_{*}\:(>1/3)$ is achieved at an intermediate value of $p_{\rm class}$. }
  \label{fig:optima-plot}
\end{figure}

\subsection{Optimization procedure}
\label{sec:sdpprocedure}

In the proof of Theorem~\ref{theorem:maintraceless} the lower bound on the approximation ratio is achieved by taking either the pure state $\rho_{\rm class}$ or the Gaussian mixed state corresponding to $\frac{1}{2}(\Gamma_{\rm quad}+\Gamma_{\rm mediator})$. This does not mean that this solution will always be the optimal Gaussian state. In fact, one can run the semi-definite program in Eq.~\eqref{eq:SDPalgorithm} ---assuming access to $\Gamma^{\rm class}$ for the moment--- and possibly get better solutions. 

The input parameter $p_{\rm class}$ of the SDP can be chosen efficiently as follows. Run the SDP in Eq. \eqref{eq:SDPalgorithm} for $p_{\rm class} = j/M$ for $j=0,1,\ldots,M = {\rm poly}(n)$. For each $j$, obtain the optimum $X^{(j)}$ of the SDP in Eq.~\eqref{eq:SDPalgorithm} and its associated covariance matrix $\Gamma^{(j)} = i\big(\mathbbm{1}-X^{(j)}\big)$. Then, calculate the expectation that $\Gamma^{(j)}$ achieves on $H$ (in Problem \ref{prob:maintracelessoptimizationproblem} \textsc{Traceless CIFH Optimization}) and pick the $\Gamma^{(j)}$ that achieves the largest expectation. Since the values $p_{\rm class} = 0,1$ are included in the sweep, the optimal $\Gamma^{(j)}$ will achieve approximation ratio at least $1/3$. In practice, we expect the optimal approximation ratio to be achieved at an intermediate value of $p_{\rm class}$ and to be larger than $1/3$. To illustrate this, Figure \ref{fig:optima-plot} gives the approximation ratio achieved by the optimum of SDP in Eq.~\eqref{eq:SDPalgorithm} on $H$ as a function of $p_{\rm class}$. Here $H$ is taken to be a $3$-site (i.e., $6$-mode) traceless Fermi-Hubbard Hamiltonian on a triangle.

Another input to the SDP in Eq.~\eqref{eq:SDPalgorithm} is $\Gamma^{\rm class}$. In the next section, we will discuss under what conditions (a sufficiently accurate approximation of) $\Gamma^{\rm class}$ can be obtained ---i.e., under what conditions our method for constructing a fermionic Gaussian state with constant approximation ratio is constructive. 

\section{Efficient approximate constructions}
\label{sec:interactionoptimization}
In this section, we discuss under what conditions one can efficiently obtain $\Gamma^{\rm class}$ (the optimum of $H_{\rm class}$) or an approximation of it -- and use it to efficiently construct a Gaussian approximation using the SDP in Eq. \eqref{eq:SDPalgorithm}. In particular, we show that if the \textit{interaction} graph $G_{\rm class}=(V,(w,E))$ in $H_{\rm class}$ in Eq.~\eqref{eq:classH} is bipartite, then $\Gamma^{\rm class}$ can efficiently be found, leading to Theorem~\ref{theorem:bipartitetraceless}. In addition, for Problem \ref{prob:maintracefuloptimizationproblem} \textsc{PSD CIFH Optimization}, we can efficiently obtain a provably accurate approximation of $\Gamma^{\rm class}$, leading to Theorem~\ref{theorem:maintraceful}. 

In Section \ref{sec:fermionicmaxcut}, we take a small detour and discuss how Gaussian approximations can be constructed using our methods for \textsc{Fermionic Max Cut} ---a fermionic version of the \textsc{Quantum Max Cut} problem \cite{gharibian_parekh}. 

In all of our constructions, when claiming polynomial-time solvability, we rely on the fact that our SDPs are of dimensionality $O(n)$ and satisfy standard conditions for solvability of semi-definite programs in polynomial time (polynomial in dimensionality, logarithm of the error, and the number of digits of precision). In particular, one can rely on Section 5.3 of the textbook by Ben-Tal and Nemirovski \cite{ben2001lectures}, which demonstrates polynomial efficiency for semi-definite programs with polynomially bounded feasible sets (see~Theorem 5.3.1 in \cite{ben2001lectures}).
For the SDP we introduced in Eq.~\eqref{eq:SDPalgorithm}, the feasible set is polynomially bounded by bounding the Frobenius norm of the matrices of spectral radius $1$ (which is required by one of the constraints). The exact same argument is sufficient for all semi-definite programs which will be introduced later in this Section.

\subsection{Proof of Theorem~\ref{theorem:bipartitetraceless}:  classical interactions on bipartite graphs}
\label{sec:proof_bipartitetraceless}

\begin{lemma}
    Let the Hamiltonian $H_{\rm class} = \sum_{(j,k)\in E}w_{j,k}\big(1/4-x_{j}x_{k} \big) + \sum_{j\in V}\mu_j \big(x_{j}-1/2\big)$ with $w_{j,k} \geq 0$, be defined on a graph $G_{\rm class} = ((\mu,V),(w,E))$ with binary variables $x_j=0,1$. If $G_{\rm class}$ is bipartite, then the classical state $\rho$, i.e. the vector ${\bf x}$, that optimizes $H_{\rm class}$ can be obtained in polynomial time. \label{lemma:interaction_algorithm_bipartite}
\end{lemma}
\begin{proof}
    Our proof directly uses a known Theorem on quadratic binary optimization problems (QUBO). In particular, the maximization problem ${\bf x}^T Q {\bf x} +{\bf c}^T {\bf x}$ over the binary vector ${\bf x}$ with $Q$ a real matrix with nonnegative off-diagonal entries (and ${\bf c}$ a real vector) can be efficiently solved: this is for example stated as Theorem~3.16 in \cite{qubo}. We can introduce Ising spin variables $z_i=1-2 x_i=\pm 1$ and rewrite $H_{\rm class}$ in terms of these variables such that the quadratic term in $H_{
    \rm class}$ equals $-\frac{1}{4}\sum_{(j,k)\in E}w_{j,k} z_{j}z_{k}$. Since $G_{\rm class}$ is bipartite with bi-partition $V=V_A \cup V_B$, applying a spin-flip $z_i \rightarrow -z_i$ for all $i\in A$, will flip the sign so that the quadratic term becomes $+\frac{1}{4}\sum_{(j,k)\in E}w_{j,k} z_{j}z_{k}$. Switching back to the $x_j$ variables thus gives nonnegative off-diagonal entries when constructing the matrix $Q$, while there are no constraints on the vector ${\bf c}$. Hence, we can efficiently obtain an optimal solution ${\bf x}$. Explicitly, the QUBO problem is rewritten as an integer linear program and relaxed to a linear program whose optimal solution can be shown to be achieved on $\{x_j\in \{0,1\}\}$, see also \cite{Schrijver}.  
\end{proof}

The procedure used in Lemma \ref{lemma:interaction_algorithm_bipartite} provides the optimum of $H_{\rm class}$ provided that $G_{\rm class}$ is a bipartite graph. For Fermi-Hubbard Hamiltonians $G_{\rm class}$ is not just bipartite, but it is simply a collection of disjoint edges. For such trivial graphs the optimum of $H_{\rm class}$ can be obtained in an extremely simple way:
\begin{remark}
    If $G_{\rm class} = ((\mu,V),(w,E))$ consists of a collection of disjoint edges, then the optimum of $H_{\rm class} = \sum_{(j,k)\in E}w_{j,k}\big( \mathbbm{1}/4-x_{j}x_{k} \big) + \sum_{j\in V}\mu_j \big(x_{j}-\mathbbm{1}/2\big)$ on $G_{\rm class}$ can be obtained using the following simple procedure. Start from the vacuum state $x_{j} = 0$ $\forall j\in V$. For each edge $(j,k) \in E$, set $x_{j} = 1$ if $\mu_{j}\geq \mu_{k}$ and set $x_{k} = 1$ if $\mu_{j}< \mu_{k}$. Then, for each edge $(j,k)\in E$, if $\min\{\mu_j,\mu_k\} \geq w_{j,k}$, also occupy the other mode ${\rm arg}\min \{\mu_{j},\mu_{k}\}$ on the edge.  
\end{remark} 

\begin{proof}[Proof of Theorem~\ref{theorem:bipartitetraceless}]
Theorem~\ref{theorem:bipartitetraceless} thus follows immediately by noting that the SDP in Eq.~\eqref{eq:SDPalgorithm}, which leads to a state with approximation ratio $1/3$, can be run in polynomial time, since its input $\Gamma^{\rm class}$ can be obtained in polynomial time. Therefore, the fermionic Gaussian state that achieves an approximation ratio at least $1/3$ for Problem \ref{prob:maintracelessoptimizationproblem} can be obtained deterministically in polynomial time. 
\end{proof}

\subsection{Proof of Theorem~\ref{theorem:maintraceful}: positive semi-definite classical interactions}
\label{sec:psd}

In this section, we obtain an approximation of the optimum $\Gamma^{\rm class}$ of $H_{\rm class}$ in Problem \ref{prob:maintracefuloptimizationproblem}. This is established in Lemma \ref{lemma:interaction_approx_algorithm}. In our proof, we make use of a result from \cite{GW}, namely:

\begin{lemma}[Theorem~3.2.1 in \cite{GW}]
Given a weighted graph $G = (V,(w,E))$ with non-negative weights $w_{j,k}\geq 0$ and a ${\rm sign}(j,k)$ for each edge $(j,k)\in E$. Consider the problem of computing 
\begin{equation}
    {\rm MaxCut_{\pm 1}} = \max_{\{z_{j}=\{\pm 1\}\}_{j\in V}}\sum_{(j,k)\in E}w_{j,k}\big( 1-{\rm sign}(j,k)z_{j}z_{k} \big).
\label{eq:signedMC}
\end{equation}
An assignment $\{z_{j}\}_{j\in V}$ that (in expectation) achieves objective value $r_{GW}{\rm MaxCut}_{\pm 1}$ (with $r_{GW} = 0.878$) for this problem can be obtained in polynomial time. 
\label{lemma:signedMC_GW}
\end{lemma} 

\begin{lemma}
    Let the Hamiltonian $H_{\rm class} = \sum_{(j,k)\in E}w_{j,k}\big( 1-x_{j}x_{k} \big) + \sum_{j\in V}\mu_j x_{j}\geq 0$ be defined on a graph $G_{\rm class} = (V,(w,E))$, and $\mu_j\geq 0$. A classical state $\rho$, i.e. the binary vector ${\bf x}$, can be efficiently found which achieves expectation $\tr\big( \rho H_{\rm class} \big) \geq r_{GW}\lambda_{\rm max}\big(H_{\rm class}\big)$. 
\label{lemma:interaction_approx_algorithm}
\end{lemma}

\begin{proof}
    We define $\tilde{\mu}_{j} = \mu_j/|E_{j}|$ $\forall j\in V$, with $|E_{j}|$ denoting the number of edges adjacent to a vertex $j$ and let $z_j=1-2 x_j$ so that we have
    \begin{equation}
       H_{\rm class}(\{z_j\})=\sum_{(j,k)\in E}\Big[ \Big(\frac{3}{4} w_{j,k}+\frac{1}{2}\tilde{\mu}_{j}+\frac{1}{2}\tilde{\mu}_{k}\Big) - \Big( \frac{1}{4}w_{j,k}-\frac{1}{2}\tilde{\mu}_{j} \Big)z_{j} - \Big( \frac{1}{4}w_{j,k}-\frac{1}{2}\tilde{\mu}_{k} \Big)z_{k} - \frac{1}{4}w_{j,k} z_{j}z_{k}\Big]. 
    \label{eq:optproblemF}
    \end{equation}
    We can relate this optimization problem to a purely quadratic optimization problem by introducing a new variable $y=\pm 1$. Consider
    \begin{multline}
        H_{\rm class}^+(\{z_j\},y) = \sum_{(j,k)\in E}\Big[ \Big(\frac{3}{4} w_{j,k}+\frac{1}{2}\tilde{\mu}_{j}+\frac{1}{2}\tilde{\mu}_{k}\Big) - \Big( \frac{1}{4}w_{j,k}-\frac{1}{2}\tilde{\mu}_{j} \Big)yz_{j}  - \Big( \frac{1}{4}w_{j,k}-\frac{1}{2}\tilde{\mu}_{k} \Big)yz_{k} - \frac{1}{4}w_{j,k} z_{j}z_{k}\Big]. 
    \label{eq:optproblemF+1}
    \end{multline}
    For an assignment $\big(y,\{z_{j}\}\big)$ achieving a value $r\, {\rm max}_{\{z_j\},y} H_{\rm class}^+$ (for some $r \leq 1$), the negated assignment $\big(\bar{y},\{\bar{z}_{j}\}\big)$ (with $\bar{z}$ denoting the spin-flip negation of $z$) clearly achieves the same value. 
    Furthermore, either the assignment $\big(y,\{z_{j}\}\big)$ (if $y = +1$) or $\big(\bar{y},\{\bar{z}_{j}\}\big)$ (if $\bar{y} = +1$) achieves the value $r \max_{\{z_j\},y} H_{\rm class}^+\geq r \max_{\{z_j\}}H_{\rm class}$ for $H_{\rm class}(\{z_j\})$ in Eq.~\eqref{eq:optproblemF}: either $\{z_{j}\}$ or $\{\bar{z}_{j}\}$ achieves the approximation ratio $r$ for $H_{\rm class}$ in Eq.~\eqref{eq:optproblemF}. Due to Lemma \ref{lemma:signedMC_GW}, we have an $r_{GW}$-approximation algorithm for the quadratic optimization problem in Eq.~\eqref{eq:optproblemF+1} which is of the form in Eq.~\eqref{eq:signedMC} plus a constant, i.e.
        \begin{multline}
        H_{\rm class}^+ = c + \sum_{(j,k)\in E}\bigg[ \biggl\lvert \frac{1}{4}w_{j,k}-\frac{1}{2}\tilde{\mu}_{j} \biggr\rvert \Big(1 - \text{sign}\Big( \frac{1}{4}w_{j,k}-\frac{1}{2}\tilde{\mu}_j \Big)yz_{j}\Big) \\ \biggl\lvert \frac{1}{4}w_{j,k}-\frac{1}{2}\tilde{\mu}_{k} \biggr\rvert \Big(1 - \text{sign}\Big( \frac{1}{4}w_{j,k}-\frac{1}{2}\tilde{\mu}_k \Big)yz_{k}\Big) +  \frac{1}{4}w_{j,k} \Big( 1-z_{j}z_{k}\Big)\bigg], 
    \end{multline}
    with constant $c \geq 0$ when $\mu_j\geq 0$. Due to $c \geq c r_{GW}$ the approximation ratio
    for the problem in Eq.~\eqref{eq:optproblemF} is thus at least $r_{GW}$. 
\end{proof}

\noindent
Now we are ready to prove Theorem \ref{theorem:maintraceful}.

\begin{proof}[Proof of Theorem~\ref{theorem:maintraceful}]
    The proof of Theorem~\ref{theorem:maintraceful} largely follows the same structure as that of Theorem~\ref{theorem:maintraceless}, so we advise the reader to first read the latter. Through Lemma \ref{lemma:interaction_approx_algorithm}, we can efficiently obtain a classical state with covariance matrix $\Gamma_{GW}^{\rm class}$ that achieves energy $r_{GW}\lambda_{\max}(H_{\rm class})$ on $H_{\rm class}$ in Problem \ref{prob:maintracefuloptimizationproblem}. We let $\tilde{\Gamma}^{\rm class}$ be the state that achieves the largest expectation on $H_{\rm class}$ out of $\Gamma_{GW}^{\rm class}$ and $-\Gamma_{GW}^{\rm class}$. 
    
    Next, let us prove how $\tilde{\Gamma}^{\rm class}$ can be used in combination with the SDP in Eq.~\eqref{eq:SDPalgorithm} to obtain a fermionic Gaussian state that achieves approximation ratio $0.637$ for Problem \ref{prob:maintracefuloptimizationproblem}. We consider the SDP in Eq.~\eqref{eq:SDPalgorithm}, where we input $\tilde{\Gamma}^{\rm class}$ instead of $\Gamma^{\rm class}$ and the cost function is shifted upwards by $\sum_{(j,k)\in E'}w'_{j,k}$. The blended Gaussian state with covariance matrix $\Gamma = p_{\rm class}\Gamma^{\rm class} + \frac{1-p_{\rm class}}{2}(\Gamma^{\rm mediator} + \Gamma^{\rm quad})$ (with $\Gamma^{\rm mediator}$ defined in Eq.~\eqref{eq:mediator}) is a feasible solution of the SDP for the same reasons as in the proof of Theorem~\ref{theorem:maintraceless}. Therefore, the optimum of the SDP achieves expectation $\sum_{(j,k)\in E'}\big(w'_{j,k}+\frac{1-p_{\rm class}}{2}\frac{1}{2}(\Gamma^{\rm quad}_{2j-1,2k}-\Gamma^{\rm quad}_{2j,2k-1})\big)$ on $H_{\rm quad}$. The optimum thus achieves expectation at least
    \begin{equation}
        \frac{1+\frac{1-p_{\rm class}}{2}}{2}\cdot \lambda_{\max}(H_{\rm quad}),
    \end{equation}
    on $H_{\rm quad}$, where we have used that $1+Cx\geq \frac{1+C}{2}(1+x)$ for $x\in [-1,+1]$ and $C \in [0,1]$, because $\frac{1+Cx}{1+x}$ decreases monotonically with $x$ for any $C \in [0,1]$.

    Next, let us argue that the approximation ratio on $H_{\rm class}$ can be lower bounded for any feasible solution of the SDP ---and therefore also for the optimum--- as follows. Using the conditions in the SDP in Eq.~\eqref{eq:SDPalgorithm}, with $\tilde{\Gamma}^{\rm class}$ as input, we conclude that any feasible solution achieves expectation at least
    \begin{equation}
        \sum_{(j,k)\in E}\frac{1}{4}w_{j,k}\Big( 3+p_{\rm class}\tilde{\Gamma}^{\rm class}_{2j-1,2j}+p_{\rm class}\tilde{\Gamma}^{\rm class}_{2k-1,2k}-p_{\rm class}^{2}\tilde{\Gamma}^{\rm class}_{2j-1,2j}\tilde{\Gamma}^{\rm class}_{2k-1,2k} \Big) + \frac{1}{2}\sum_{j\in V}\mu_{j}\big( 1-p_{\rm class}\tilde{\Gamma}^{\rm class}_{2j-1,2j} \big)
    \end{equation}
    on $H_{\rm class}$. Using reasoning similar as in Lemma \ref{lemma:diagonalcostfunctionKbehaviour}, we argue the following. 
    \begin{align}
        \sum_{(j,k)\in E}&\: \frac{1}{4}w_{j,k}\Big( p_{\rm class}\tilde{\Gamma}^{\rm class}_{2j-1,2j}+p_{\rm class}\tilde{\Gamma}^{\rm class}_{2k-1,2k}-p_{\rm class}^{2}\tilde{\Gamma}^{\rm class}_{2j-1,2j}\tilde{\Gamma}^{\rm class}_{2k-1,2k} \Big) - \frac{1}{2}\sum_{j\in V}\mu_{j}p_{\rm class}\tilde{\Gamma}^{\rm class}_{2j-1,2j} \nonumber \\ &\: \geq p_{\rm class}^2\bigg[\sum_{(j,k)\in E}\frac{1}{4}\Big( \tilde{\Gamma}^{\rm class}_{2j-1,2j}+\tilde{\Gamma}^{\rm class}_{2k-1,2k}-\tilde{\Gamma}^{\rm class}_{2j-1,2j}\tilde{\Gamma}^{\rm class}_{2k-1,2k} \Big) - \frac{1}{2}\sum_{j\in V}\mu_{j}\tilde{\Gamma}^{\rm class}_{2j-1,2j}\bigg].
    \end{align}
    This follows from the fact that the sum $W_{\rm lin}$ of the contributions \textit{linear} in entries of $\tilde{\Gamma}^{\rm class}$ is non-negative and so $p_{\rm class}W_{\rm lin} \geq p_{\rm class}^{2}W_{\rm lin}$. Else, $\tilde{\Gamma}^{\rm class}$ would have been chosen to equal $\Gamma_{GW}^{\rm class}$ with opposite sign. Hence any feasible solution achieves at least the following expectation on $H_{\rm class}$. 
    \begin{align}
        &\sum_{(j,k)\in E} \frac{1}{4}w_{j,k}\Big( 3+p_{\rm class}^2\big(\tilde{\Gamma}^{\rm class}_{2j-1,2j}+\tilde{\Gamma}^{\rm class}_{2k-1,2k}-\tilde{\Gamma}^{\rm class}_{2j-1,2j}\tilde{\Gamma}^{\rm class}_{2k-1,2k}\big) \Big) + \frac{1}{2}\sum_{j\in V}\mu_{j}\big( 1-p_{\rm class}^2\tilde{\Gamma}^{\rm class}_{2j-1,2j} \big) \nonumber \\ &\: \geq \frac{3+p_{\rm class}^2}{4}\sum_{(j,k)\in E}\frac{1}{4}w_{j,k}\Big( 3+\big(\tilde{\Gamma}^{\rm class}_{2j-1,2j}+\tilde{\Gamma}^{\rm class}_{2k-1,2k}-\tilde{\Gamma}^{\rm class}_{2j-1,2j}\tilde{\Gamma}^{\rm class}_{2k-1,2k}\big) \Big) + \frac{1+p_{\rm class}^2}{2}\frac{1}{2}\sum_{j\in V}\mu_{j}\big( 1-\tilde{\Gamma}^{\rm class}_{2j-1,2j} \big) \nonumber \\ &\: \geq \frac{1+p^2_{\rm class}}{2}r_{GW}\lambda_{\max}(H_{\rm class}),
    \end{align}
    where we have used that $3+Cx\geq \frac{3+C}{4}(3+x)$ for $x\in [-3,+1]$ and $C\in [0,1]$, that $1-Cx\geq \frac{1+C}{2}(1-x)$ for $x\in [-1,+1]$ and $C\in [0,1]$, and that $\big(3+p_{\rm class}^{2}\big)/4 \geq \big( 1+p_{\rm class}^2 \big)/2$. 

    Therefore, the optimum of the SDP in Eq.~\eqref{eq:SDPalgorithm} (with $\tilde{\Gamma}^{\rm class}$ as input) achieves approximation ratio on $H$ in Problem \ref{prob:maintracefuloptimizationproblem} which is at least 
    \begin{equation}
        \frac{\big((1+p^2_{\rm class})/2\big)\:r_{GW}\beta + \big(1/2+(1-p_{\rm class})/4\big)}{\beta + 1} =: f_{\beta}(p_{\rm class}),
    \end{equation}
    with $\beta = \lambda_{\max}(H_{\rm class})/\lambda_{\max}(H_{\rm quad})\geq 0$. The function $f_{\beta}(p_{\rm class})$ is convex and for given $\beta$ the optimum is achieved at $p_{\rm class} = 0$ or $p_{\rm class} = 1$. At $\beta = 2r_{GW}$, $f_{\beta}(p_{\rm class} = 0) = f_{\beta}(p_{\rm class} = 1) = r_{GW}/(r_{GW}+1/2)$ while at other values of $\beta$, $\max_{p_{\rm class} = 0,1}f_{\beta}(p_{\rm class}) \geq r_{GW}/(r_{GW}+1/2)$. Hence, we obtain the lower bound of $r_{GW}/(r_{GW}+1/2) \geq 0.637$ from the theorem statement. We note that to efficiently obtain a fermionic Gaussian state achieving at least this approximation ratio, one has to optimize over $p_{\rm class}$ as discussed in Section \ref{sec:sdpprocedure}.
\end{proof} 

\subsection{Fermionic Max Cut}
\label{sec:fermionicmaxcut}

Inspired by \textsc{Quantum Max Cut} \cite{gharibian_parekh}, we introduce another model with positive semi-definite terms, namely:
\begin{problem}[\textsc{Fermionic Max Cut}]
     Consider the Hamiltonian 
         \begin{align}
             H = \sum_{(j,k)\in E}\frac{1}{2}w_{j,k}\left(-a_j^{\dagger}a_k-a_k^{\dagger}a_j+n_j+n_k-2 n_j n_k\right)
             =\notag \\
             \sum_{(j,k)\in E}\frac{1}{4}w_{j,k}\big(  \mathbbm{1} + ic_{2j-1}c_{2k} - ic_{2j}c_{2k-1} + c_{2j-1}c_{2j}c_{2k-1}c_{2k} \big)=
             \notag \\
          \sum_{(j,k)\in E}w_{j,k}\left(\frac{\mathbbm{1}+ic_{2j-1}c_{2k}}{2}\right) \left(\frac{ \mathbbm{1}-ic_{2j}c_{2k-1}}{2}\right)   
             \succeq 0,
             \label{eq:FMC}
         \end{align}
     with $w_{j,k}\geq 0$, and vertex set $V$ and edge set $E$. Compute $\lambda_{\max}(H) = \max_{\rho \in \mathbb{C}^{2^n\times 2^n}} \Big\{ \tr \big( \rho H \big) \text{ s.t. } \rho \succeq 0, \: \tr(\rho) = 1 \Big\}$.
\label{prob:tracefulfermionicmaxcut}
\end{problem}
Note that $H$ in Eq.~\eqref{eq:FMC} is manifestly positive semi-definite since each term is a projector onto a pure Gaussian 2-mode state with $\Gamma_{2j-1,2k} = +1,\: \Gamma_{2j,2k-1} = -1$. When put on a line, by the Jordan-Wigner transformation, this model is equivalent to 
\textsc{Quantum Max Cut}, where the projector on each edge projects onto a 2-qubit singlet state. For more general graphs, \textsc{Fermionic Max Cut} does not have the $U^{\otimes n}$ symmetry of \textsc{Quantum Max Cut} (aka the anti-ferromagnetic Heisenberg model) which allows it to be more easily solvable/approximable \cite{huber2024}. We introduce \textsc{Fermionic Max Cut} as a novel generalization of the classical Max Cut problem which may be more amenable to approximate optimization methods than the general Problem
\ref{prob:maintracefuloptimizationproblem}.

Indeed, observe that \textsc{Fermionic Max Cut} is like Problem \ref{prob:maintracefuloptimizationproblem} with $\mu_j=\sum_{k\in E_{j}}\frac{1}{2}w_{j,k}$ $\forall j$ (with $E_{j}$ denoting the subset of edges involving mode $j$), with the edge sets  coinciding, i.e. $(w',E') = (\frac{1}{2}w,E)$, and only requiring that the \textit{sum} of the classical interaction and hopping interactions is positive semi-definite. Note, however, that in that parameter regime, the trace of $H$ in Problem \ref{prob:maintracefuloptimizationproblem} is larger than that of $H$ in \textsc{Fermionic Max Cut} by an amount $\sum_{j,k\in E}\frac{9}{8}w_{j,k}\tr(\mathbbm{1})$. Hence, values for approximation ratios for Problem \ref{prob:maintracefuloptimizationproblem} correspond to different values of approximation ratios for \textsc{Fermionic Max Cut}. Observe, for instance, that the maximally-mixed state achieves approximation ratio $1/4$ for \textsc{Fermionic Max Cut}. It is straightforward to prove the following proposition. 

\begin{proposition}
    There exists a fermionic Gaussian state $\rho$ that achieves approximation ratio $\frac{1}{2}$ for \textsc{Fermionic Max Cut}, and this state can be obtained deterministically in polynomial time. 
\label{theorem:tracefulFMC}
\end{proposition}

\begin{proof}
    Let $H_{\rm class} = \sum_{(j,k)\in E}\frac{1}{4}w_{j,k}\big(  \mathbbm{1} + c_{2j-1}c_{2j}c_{2k-1}c_{2k} \big)$ and $H_{\rm quad} = \sum_{(j,k)\in E}\frac{1}{4}w_{j,k}\big(  ic_{2j-1}c_{2k} - ic_{2j}c_{2k-1} \big)$ and let $\Gamma^{\rm quad}$ denote the covariance matrix of the optimum of $H_{\rm quad}$. Let $\Gamma^{\rm mediator}$ be defined as in Eq.~\eqref{eq:mediator}. The Gaussian blend $\Gamma = \frac{1}{2}\big( \Gamma^{\rm mediator} + \Gamma^{\rm quad} \big)$ (see Lemma \ref{lemma:recombinationofgaussianstates}) with density matrix $\rho_{\rm Gauss}$ achieves expectation 
    \begin{multline}
        {\rm Tr}(\rho_{\rm Gauss} H_{\rm class})=\sum_{j,k\in E}\frac{1}{4}w_{j,k}\big( 1 - \Gamma_{2j-1,2j}\Gamma_{2k-1,2k} - \Gamma_{2j-1,2k}\Gamma_{2j,2k-1} + \Gamma_{2j-1,2k-1}\Gamma_{2j,2k}\big) = \\ \sum_{j,k\in E}\frac{1}{4}w_{j,k}\Big( 1 - \frac{1}{4}\Gamma^{\rm quad}_{2j-1,2k}\Gamma^{\rm quad}_{2j,2k-1} + \frac{1}{4}\Gamma^{\rm quad}_{2j-1,2k-1}\Gamma^{\rm quad}_{2j,2k}\Big) \geq \sum_{(j,k)\in E}\frac{1}{4}w_{j,k}\geq \frac{1}{2}\lambda_{\max}(H_{\rm class}),
    \end{multline}
    where we have used Proposition \ref{prop:Wickstheorem} and the fact that $\Gamma^{\rm quad}$ corresponds to a Slater determinant state, see Lemma \ref{lemma:Slaterexpectations}. Since ${\rm Tr}(\rho_{\rm Gauss}H_{\rm quad})=\frac{1}{2}\lambda_{\max}(H_{\rm quad})$ on $H_{\rm quad}$, the approximation of $\rho_{\rm Gauss}$ is thus at least ${\rm Tr}(\rho_{\rm Gauss}H)/\lambda_{\max}(H)\geq 1/2$ , using $\lambda_{\max}(H) \leq \lambda_{\max}(H_{\rm class}) + \lambda_{\max}(H_{\rm quad})$. Through Remark \ref{remark:mixedgaussianismixturepuregaussians}, there is a pure fermionic Gaussian state that achieves approximation ratio at least $\frac{1}{2}$. Since $\Gamma^{\rm quad}$ and $\Gamma^{\rm mediator}$ can be obtained in polynomial time, the Gaussian blend can be obtained efficiently. 
\end{proof}

Fermionic Gaussian states thus achieve approximation ratio at least $\frac{1}{2}$ on \textsc{Fermionic Max Cut}. This can be contrasted with the fact that product states achieve approximation ratio \textit{at most} $\frac{1}{2}$ on \textsc{Quantum Max Cut} \cite{gharibian_parekh}. 

\section{Fermionic optimization in the presence of a particle constraint}
\label{sec:q_particle_proof}

Here, we consider approximation algorithms in case the optimization problem involves an average particle number constraint, see Problem \ref{prob:maintracelessoptimizationproblem_averagehalffilling}. We prove that one can still obtain an approximation ratio of $1/\big[2\big( (n-2q)/n + 3/2 \big)\big]$ (reducing to $1/3$ at half-filling, i.e., for $q = n/2$) and an efficient algorithm to find such Gaussian state in case of a bipartite classical interaction graph when fixing the average particle number to $q$, see Theorem~\ref{theorem:maintraceless_averagehalffilling}.

Let us denote the covariance matrix of the optimum of $H_{\rm class}$ in Eq.~\eqref{eq:classH} (here with $\mu_j = 0$ $\forall j$) at \textit{exactly} $q$ particles 
by $\Gamma^{\rm class}_q$, and its associated expectation by $\lambda_{\max,q}(H_{\rm class})$. 
We denote the covariance matrix of the \textit{overall} optimum of $H_{\rm quad}$ by $\Gamma^{\rm quad}$.
To prove Theorem~\ref{theorem:maintraceless_averagehalffilling}, let us first establish the following two simple lemmas. Lemma \ref{lemma:interactionaveragehalffilling} says that the optimum of $H_{\rm class}$ at \textit{average} particle number $q \in \{0,1,\ldots,\lfloor n/2 \rfloor\}$ is in fact a classical state at particle number $q$ and can be obtained efficiently. Lemma~\ref{lemma:Hclassfacts_fixedparticlenumber} establishes two facts about $H_{\rm class}$ that we will use in the proof of Theorem~\ref{theorem:maintraceless_averagehalffilling}. We only prove these two lemmas for $q\leq \lfloor n/2 \rfloor$, hence Theorem~\ref{theorem:maintraceless_averagehalffilling} is only proved for those $q$'s. 

\begin{lemma}
    Given a bipartite interaction graph $G_{\rm class} = (V,(w,E))$. For each $q \in \{0,1,\ldots,\lfloor n/2 \rfloor\}$, the optimal classical state $\Gamma^{\rm class}_{q}$ (see Definition \ref{def:Slaterstates}) at particle number $q$ achieves the average-$q$ optimum 
    $\lambda_{\max,\langle q \rangle}(H_{\rm class}) = \max_{\rho\in \mathbb{C}^{2^n \times 2^n}} \Big\{ \tr\big( \rho H_{\rm class} \big) \text{ s.t. } \:  \tr \big(\rho \hat{N}\big) = q, \: \rho \succeq 0, \: \tr(\rho) = 1 \Big\}$. This state can be obtained in polynomial time.
\label{lemma:interactionaveragehalffilling}
\end{lemma}
\begin{proof}
    The graph $G_{\rm class} = (V,(w,E))$ is bipartite w.r.t. a bi-partition $V = V_{A}\cup V_{B}$, where w.l.o.g. we take $|V_{A}|\geq |V_{B}|$ so that $|V_{A}|\geq \lfloor n/2 \rfloor$. Therefore, the classical state $\prod_{j\in V_{A}}a_{j}^{\dagger}\ket{\rm vac}$ (for which $n_{j} = 1$ for all $j\in V_{A}$ and $n_{k} = 0$ for all $k\in V_{B}$) is a state at particle number $|V_{A}|\geq q$ and achieves expectation $\frac{1}{4}\sum_{(j,k)\in E}w_{j,k}$. Clearly, one can annihilate particles in modes $j\in V_{A}$ until a classical state at particle number $q$ is achieved, while preserving the expectation $\frac{1}{4}\sum_{(j,k)\in E}w_{j,k}$. Since the optimum $\lambda_{\max,\langle q \rangle}(H_{\rm class})$ is at most $\frac{1}{4}\sum_{(j,k)\in E}w_{j,k}$, the lemma statement follows. 
\end{proof}

\begin{lemma}
    Given a bipartite interaction graph $G_{\rm class} = (V,(w,E))$. 
    Let us define $F(z_1,\ldots,z_n) = -\frac{1}{4}\sum_{j,k\in E} w_{j,k}(z_j + z_k + z_j z_k)$ with $\{z_{j} = \pm 1\}_{j\in V}$. Let $z_1^q,\ldots,z_n^q$ denote the optimal assignment with exactly $q$ variables set to $+1$, and $F_q = F(z_1^q,\ldots,z_n^q)$. Then, (1) $F_{0} = F_{1} = F_{2} = \ldots = F_{\lfloor n/2 \rfloor}$ and (2) $F(p\:z_1^q,\ldots,p\:z_n^q)\geq p^2 F(z_1^q,\ldots,z_n^q) = p^2 F_{q}$ for $p\in [0,1]$ and $0\leq q\leq \lfloor n/2 \rfloor$. 
\label{lemma:Hclassfacts_fixedparticlenumber}
\end{lemma}
\begin{proof}
    As argued in the proof of Lemma \ref{lemma:interactionaveragehalffilling}, the optima at $0\leq q\leq \lfloor n/2 \rfloor$ are equal to $\frac{1}{4}\sum_{(j,k)\in E}w_{j,k}$ so that $F_{1} = F_{2} = \ldots = F_{\lfloor n/2 \rfloor}$. The optimum at each $0 \leq q \leq \lfloor n/2 \rfloor$ w.l.o.g. is such that for each edge $(j,k)\in E$, we have that $(z^q_{j} = -1, z^q_{k} = +1)$, $(z^q_{j} = +1, z^q_{k} = -1)$ or $(z^q_{j} = -1, z^q_{k} = -1)$. Let $F_1(z_1,\ldots,z_n) = -\frac{1}{4}\sum_{j,k\in E} w_{j,k}(z_j + z_k)$ and $F_2(z_1,\ldots,z_n) = -\frac{1}{4}\sum_{j,k\in E} w_{j,k}z_j z_k$. Then clearly, for each $0\leq q\leq \lfloor n/2 \rfloor$, $F_1(z^{q}_{1},\ldots,z^{q}_{n})\geq 0$. Therefore, $F_1(p\:z^{q}_{1},\ldots,p\:z^{q}_{n})\geq p^2 F_1(z^{q}_{1},\ldots,z^{q}_{n})$ for $p\in [0,1]$. Since $F_2(p\:z_1,\ldots,p\:z_n) = p^2 F_2(z_1,\ldots,z_n)$, we have that $F(p\:z_1^q,\ldots,p\:z_n^q)\geq p^2 F(z_1^q,\ldots,z_n^q)$ for $p\in [0,1]$ and $0\leq q\leq \lfloor n/2 \rfloor$. 
\end{proof}

Having established these two facts, let us prove Theorem~\ref{theorem:maintraceless_averagehalffilling}, which follows the same structure as the proof of Theorem~\ref{theorem:maintraceless}. We advise the reader to first read the latter. 
\begin{proof}[Proof of Theorem~\ref{theorem:maintraceless_averagehalffilling}]
    Let us consider the SDP in Eq. \eqref{eq:SDPalgorithm}, which we alter in two ways. We take as input $\Gamma^{\rm class}_{q'}$ for some $0\leq q'\leq q$ that will be specified later in this proof, and we add the linear constraint $\sum_{j = 1}^{n}(-iX_{2j-1,2j}) = n-2q$ to the SDP. Clearly, the resulting SDP is still an instance of Lemma \ref{lemma:constrainedquadraticoptimization}. 
    
    Let us define $\Gamma^{\rm mediator}$ as in Eq.~ \eqref{eq:mediator}, given $\Gamma^{\rm quad}$. The Gaussian blend $-iX_{\rm blend} = p_{\rm class}\Gamma^{\rm class}_{q'} + \frac{1-p_{\rm class}}{2}(\Gamma^{\rm mediator} + \Gamma^{\rm quad})$ is a feasible solution to the SDP ---provided that $p_{\rm class} = \frac{n-2q}{n-2q'}$. To see this, note that 
    \begin{enumerate}
        \item $-i(X_{\rm blend})_{2j-1,2j} = p_{\rm class}(\Gamma^{\rm class}_{q'})_{2j-1,2j}$, $\forall j\in V$. 
        \item $\Gamma^{\rm quad}$ is a covariance matrix of a Slater determinant state. 
        Therefore, we have $(X_{\rm blend})_{2j-1,2k} = -(X_{\rm blend})_{2j,2k-1}$ and $(X_{\rm blend})_{2j-1,2k-1} = (X_{\rm blend})_{2j,2k}$ for all $(j,k)\in E$ (see Lemma \ref{lemma:Slaterexpectations}). 
        \item $\sum_{j = 1}^{n}\big(-i(X_{\rm blend})_{2j-1,2j}\big) = p_{\rm class}\sum_{j = 1}^{n}\big(\Gamma_{q'}^{\rm class}\big)_{2j-1,2j} = p_{\rm class}(n-2q') = n-2q$.
    \end{enumerate}
    This feasible solution ---and therefore the optimum--- achieves approximation ratio at least $\frac{1-p_{\rm class}}{2}$ of $\lambda_{\max}(H_{\rm quad})$, and therefore also of $\lambda_{\max,\langle q \rangle}(H_{\rm quad})$ (since $\lambda_{\max}(H_{\rm quad}) \geq \lambda_{\max,\langle q \rangle}(H_{\rm quad})$). What is left is to prove that all feasible solutions of the SDP ---and so also its optimum--- achieve expectation at least $p_{\rm class}^2\lambda_{\max,\langle q \rangle}(H_{\rm class})$ on $H_{\rm class}$. Through Lemma \ref{lemma:interactionaveragehalffilling}, this is equivalent to showing that all feasible solutions achieve at least expectation $p_{\rm class}^2 \lambda_{\max,q}(H_{\rm class})$ (i.e., the \textit{exact} $q$-particle optimum) on $H_{\rm class}$. Since $\lambda_{\max,q}(H_{\rm class}) = \lambda_{\max,q'}(H_{\rm class})$ for $q'\leq q\leq \lfloor n/2 \rfloor$ (see Lemma \ref{lemma:Hclassfacts_fixedparticlenumber}), it suffices to show that all feasible solutions achieve at least expectation $p_{\rm class}^2 \lambda_{\max,q'}(H_{\rm class})$. This in turn follows from the SDP constraints $-i(X_{\rm blend})_{2j-1,2j} = p_{\rm class}\Gamma^{\rm class}_{q'}$ $\forall j\in V$ in combination with Lemma \ref{lemma:Hclassfacts_fixedparticlenumber}. 

    Therefore, the optimum of the SDP achieves approximation ratio at least 
    \begin{equation}
        \frac{p_{\rm class}^2\lambda_{\max,\langle q \rangle}(H_{\rm class}) + \frac{1-p_{\rm class}}{2}\lambda_{\max,\langle q \rangle}(H_{\rm quad})}{\lambda_{\max,\langle q \rangle}(H)} \geq \frac{p_{\rm class}^2\beta + \frac{1-p_{\rm class}}{2}}{\beta + 1} = \frac{\Big(\frac{n-2q}{n-2q'}\Big)^2\beta + \frac{1}{2}-\frac{n-2q}{2(n-2q')}}{\beta + 1} = f_{\beta,q'}(p_{\rm class}),
    \end{equation}
    with $\beta := \lambda_{\max,\langle q \rangle}(H_{\rm class})/\lambda_{\max,\langle q \rangle}(H_{\rm quad})$ and we have used $\lambda_{\max,\langle q \rangle}(H) \leq \lambda_{\max,\langle q \rangle}(H_{\rm class}) + \lambda_{\max,\langle q \rangle}(H_{\rm quad})$. The function $f_{\beta,q'}(p_{\rm class})$ is convex and so its optimum is achieved at the boundaries of its domain; at $p_{\rm class} = \frac{n-2q}{n}$ (at $q' = 0$) or at $p_{\rm class} = 1$ (at $q' = q$). For each $\beta$, the approximation ratio can then be shown to be lower bounded by $1/\big[2\big( (n-2q)/n+3/2 \big)\big]$. Through Remark \ref{remark:mixedgaussianismixturepuregaussians}, there is a pure fermionic Gaussian state that achieves this approximation ratio. Using Lemma \ref{lemma:interactionaveragehalffilling}, we infer that a fermionic Gaussian state at that approximation ratio can also be efficiently constructed, because $\Gamma^{\rm class}_{q'}$ can be efficiently obtained. 
\end{proof}
Note that to obtain the approximation ratio $1/\big[2\big( (n-2q)/n+3/2 \big)\big]$ in this proof, we used a Gaussian blend feasible solution at $p_{\rm class} = \frac{n-2q}{n}$ and $p_{\rm class} = 1$, where the first choice for $p_{\rm class}$ constitutes a genuine three-component blend.

\vspace{-0.3cm}
\section{Discussion}

This work deals with the optimization of classically interacting fermion Hamiltonians --- a problem directly motivated by quantum chemistry and condensed matter theory. We give several guarantees on approximating ground energy of such Hamiltonians using Gaussian fermionic states. In particular, we show that traceless classically interacting Hamiltonians admit constant-ratio Gaussian approximations --- ruling out the Gaussian breakdown scenario, previously found for SYK-like models. Furthermore, we provide efficient constructions for Gaussian approximations to several traceless and positive semi-definite fermionic Hamiltonians, allowing also the inclusion of a particle number constraint. Our results are derived using the new notion of a Gaussian blend, allowing to construct Gaussian states with desired properties using mixtures of covariance matrices. Another technical contribution is a semi-definite program for Gaussian optimization of classically interacting Hamiltonians, which may be an interesting object for further analysis. 
On the practical side, our results help to build a rigorous basis for the Hartree-Fock approach, a standard heuristic in computational quantum chemistry and materials science. 

On an intuitive level, the reason behind classically interacting Hamiltonians avoiding the Gaussian breakdown of SYK-like models is the fact that the interactions here are commuting. The widespread non-commutation as the reason for Gaussian breakdown has been most directly pinpointed in \cite{anschuetz2024}, which gave circuit size lower bounds for SYK ground state approximations using the so-called \textit{commutation index}. For non-classical particle number conserving Hamiltonians, the Gaussian states may not be guaranteed to yield a constant-ratio approximation. A more detailed analysis of SYK-like models with particle conservation, perhaps also employing the commutation index, is an interesting subject for future study.

Since the Hamiltonians in this paper commute with the particle number operator, one might wonder whether Slater determinant states (i.e., eigenstates of particle number conserving free-fermion Hamiltonians, see Definition \ref{def:Slaterstates}) could perform as good as fermionic Gaussian states in the optimization procedures. Note, however, that the fact that a Hamiltonian is particle number conserving does not necessarily translate to Slater determinants being good approximations to the ground state. Indeed, there exist such Hamiltonians \cite{BGKT:manybody} (e.g. Richardson's Hamiltonian) for which Slater determinant states and Gaussian states asymptotically achieve approximation ratios $0$ and $1$, respectively. 

To obtain the results in this work, we have used Gaussian blends that are \textit{perfectly mediated} -- i.e., $\Gamma^{\rm quad}$ and $\Gamma^{\rm mediator}$ are blended with equal weight. One might wonder whether our results can be improved by implementing \textit{im}perfect mediation. If the weight of $\Gamma^{\rm quad}$ were larger than the weight of $\Gamma^{\rm mediator}$, then the entries $\Gamma^{\rm quad}_{2j-1,2j}$ would contribute to the expectation on $H_{\rm class}$. 
One may generally have to assume that $\Gamma^{\rm quad}$ gives a contribution to the expectation on $H_{\rm class}$ that would scale like its minimum eigenvalue $\lambda_{\min}(H_{\rm class})$. Interestingly, as we argue in Appendix \ref{sec:minvsmaxscalingclassicalcostfunction} for Problem \ref{prob:maintracelessoptimizationproblem} of traceless fermionic optimization, $|\lambda_{\min}(H_{\rm class})|$ can scale as $n \lambda_{\max}(H_{\rm class})$ in dense cases, making it difficult to obtain an improved lower bound on the approximation ratio using imperfect mediation. A regime in which imperfect mediation would be particularly useful is in the weakly interacting regime, since ideally one would obtain an approximation ratio equal to $1$ in the limit of vanishing interactions. When implementing perfect mediation, however, the approximation ratio in that regime is at most $\frac{1}{2}$. 

One may anticipate that these first values of approximation ratios can be improved by considering different optimization strategies, such as those that have been used for \textsc{Quantum Max Cut} \cite{PT:hierarchy,gharibian_parekh,anshugossetmorenz,King2023improved,huber2024}.
In particular, in Appendix \ref{sec:SDPrelaxandround} we provide the SDP relaxation hierarchy (Lasserre hierarchy) of the problem of optimizing fermionic Hamiltonians $H$ over {\em general} fermionic states. It is an open question to what extent one can use the solution of such an SDP, which is not a physical state, to round to a Gaussian state with $\Gamma$ approximately optimizing $F(\Gamma)$ in Eq. \eqref{eq:F}, for positive semi-definite Hamiltonians with classical interactions (Problem \ref{prob:maintracefuloptimizationproblem}) or specifically \textsc{Fermionic Max Cut} in Section \ref{sec:fermionicmaxcut}. 

Another open question pertains to the $H_{\rm quad}$ contributions to the Hamiltonians in this work. Our proof techniques use that they are particle number conserving. Whether the same results can be obtained if $H_{\rm quad}$ is just parity preserving is currently not known to us.

\section{Acknowledgments}
This work is supported by QuTech NWO funding 2020-2028 – Part I “Fundamental Research”, project number 601.QT.001-1, financed by the Dutch Research Council (NWO). This work was done in part while MES and YH were visiting the Simons Institute for the Theory of Computing, Berkeley. Research of YH at Perimeter Institute is supported in part by the Government of Canada through the Department of Innovation, Science and Economic Development and by the Province of Ontario through the Ministry of Colleges and Universities. We thank Anthony Chen, David Gosset, and Jonas Helsen for discussions.


\pagebreak
\appendix
\section{QMA-hardness of optimizing classically interacting fermions}
\label{sec:QMA}

Adapting Theorem~3 in \cite{QMA-hardness_FH}, one can show

\begin{corollary}
\label{thm:hubbard_QMA}
Consider a system of $2n$ fermionic modes, with fermionic operators $\{a_{i, \sigma}^{\dagger},a_{i, \sigma}\}_{i\in[n], \sigma\in \{\pm 1\}}$. There exist constants $p > q > 0$ such that for all $u \geq n^{14 + 3 p + 2q}$, determining to precision $n^{-q}$ the ground state energy\footnote{In this section, in keeping with \cite{QMA-hardness_FH} and without loss of generality for our purposes, we refer to the \textit{smallest} eigenvalue as the ground energy.} of a Hamiltonian
\begin{equation}
H
=
u \sum_{i \in [n]}
n_{i, +1} n_{i, -1}
+
\sum_{i < j,\sigma\in\{\pm1\}}
t_{i, j, \sigma }
\left(
a_{i, \sigma}^{\dagger}
a_{j, \sigma}
+
a_{j, \sigma}^{\dagger}
a_{i, \sigma}
\right)
-\mu \sum_{i \in [n],\sigma\in\{\pm 1\}}
n_{i,\sigma}
\label{eq:Hubbard_QMA}
\end{equation}
subject to $\left|t_{i, j,\sigma}\right| \leq \sqrt{n^p u}$, $\mu$ and $u/10\geq \mu\geq 10\cdot n^2\cdot \sqrt{n^p u}$ 
is QMA-complete.
\end{corollary}
\begin{proof}

For convenience, let us split the Hamiltonian into terms $H=H_u+H_t+H_\mu$, defined in a straightforward way based on the form of Eq.~\eqref{eq:Hubbard_QMA}.

    The statement only differs from Theorem~3 in \cite{QMA-hardness_FH} by two points. First, we included an additional chemical potential term $H_\mu$ into the Hamiltonian. And second, we are interested in the ground energy of the Hamiltonian as a whole, while \cite{QMA-hardness_FH} considered the Hamiltonian projected onto the $n$-particle (Hamming weight $n$) subspace. Our goal will be to show that including the chemical potential term, the ground state of the Hamiltonian in Eq.~\eqref{eq:Hubbard_QMA} is guaranteed to have $n$ particles. This would be sufficient for the statement to follow from \cite{QMA-hardness_FH} directly, because within the $n$-particle subspace where $H_\mu$ is constant, the search of the ground energy of $H=H_u+H_t$ is equivalent to that of $H=H_u+H_t+H_\mu$.

    First, let us show that the ground state cannot have more than $n$ particles. Consider a general state $\rho$ with $n$ particles, such that for all $i$, either $\mathrm{tr}(\rho\,n_{i,1})=1$ or $\mathrm{tr}(\rho\,n_{i,-1})=1$ (and thus $\mathrm{tr}(\rho H_u)=0$). Any such state has lower energy than any state $\rho'$ with $n'>n$ particles, because (using a triangle inequality on $H_t$)
    \begin{align}
        \mathrm{tr}\left(\rho' H\right)&\geq u(n'-n)-|\lambda_{\rm min}(H_t)|-\mu n'\geq u(n'-n)-2n^2\sqrt{n^p u}-\mu n', \\
        \mathrm{tr}\left(\rho H\right)&\leq \lambda_{\rm max}(H_t)- \mu n \leq  2n^2\sqrt{n^p u}-\mu n, \label{eq:QMA_proof_phi}
    \end{align}
    and $u(n'-n)>\mu(n'-n)+4n^2\sqrt{n^p u}$ for large enough $n$, given that $(n'-n)\geq1$ and the assumptions on $u$ and $\mu$.

    Secondly, the ground state cannot have less than $n$ particles. Indeed, for any state $\rho'$ with $n'<n$ particles, its energy is:
    \begin{align}
        \mathrm{tr}\left(\rho' H\right)&\geq -|\lambda_{\rm min}(H_t)|-\mu n'\geq -2n^2\sqrt{n^p u}-\mu n'.
    \end{align}
    Comparing to Eq.~\eqref{eq:QMA_proof_phi}, we see that $\rho'$ has energy greater than any state $\rho$, because $\mu(n-n')\geq4n^2\sqrt{n^p u}$ due to $(n-n')\geq 1$ and the assumptions on $\mu$. This concludes the proof.
\end{proof}


\section{Upper bound on the Gaussian approximation ratio}
\label{sec:upper-bound}

We give a small $n=4$ example of a fermionic Hamiltonian ---mapping to the anti-ferromagnetic Heisenberg model on a line of 4 qubits--- which has a unique non-Gaussian maximum eigenstate, allowing to bound the Gaussian approximation ratio in this instance. Note that for $n<4$, all states are Gaussian \cite{TerhalNoisyFermQC}. To our knowledge, this result is the first rigorous upper bound for a Gaussian approximation ratio for any classically interacting fermionic Hamiltonian.

\begin{proposition}
    For any fermionic Gaussian state $\rho_{\rm Gauss}$, one has $\tr(\rho_{\rm Gauss} H)/\lambda_{\max}(H) \leq 0.99904$, with $H$ an instance of Problem \ref{prob:maintracelessoptimizationproblem} \textsc{Traceless CIFH Optimization} for $n=4$. 
\label{prop:gaussianupperbound}
\end{proposition}
\begin{proof}
We consider the following $4$-mode Hamiltonian on a line:
\begin{equation}
    H = -\frac{1}{2}\sum_{i=1}^{3}\big( a_{i}^{\dagger}a_{i+1} + a_{i+1}^{\dagger}a_{i} \big) + \sum_{i=1}^{3}\big( \mathbbm{1}/4 - n_{i}n_{i+1} \big) + \frac{1}{2}\big(n_1 - \mathbbm{1}/2\big) + \sum_{i=2}^{3}\big(n_i - \mathbbm{1}/2\big) + \frac{1}{2}\big(n_4 - \mathbbm{1}/2\big). 
\label{eq:Hgaussianupperbound}
\end{equation}
This Hamiltonian maps onto $H = -\frac{1}{4}\sum_{i=1}^{3}\big( X_{i}X_{i+1} + Y_{i}Y_{i+1} + Z_{i}Z_{i+1} \big)$ under the Jordan-Wigner transformation, hence its maximum eigenstate corresponds to the ground state of the anti-ferromagnetic Heisenberg model on a line. For the remainder of the argument, we need the known values of $\lambda_{\max}(H)$, the spectral gap $\Delta := \lambda_{\max}(H) - \lambda_{\max - 1}(H)$, and the maximum energy eigenstate $\ket{\psi_{\rm max}}$. Due to the Lieb-Mattis Theorem \cite{LiebMattis} $H$ has a \textit{unique} eigenstate at $\lambda_{\max}(H)$. One can find that $\lambda_{\max}(H) = \frac{1}{4}(3 + 2\sqrt{3})$ and $\Delta = \frac{1}{2}(1+\sqrt{3}-\sqrt{2})$. For the unique maximum energy eigenstate $\ket{\psi_{\max}}$, one can show that its fermionic covariance matrix obeys $\Gamma_{\psi_{\max}}\Gamma_{\psi_{\max}}^{T} = s\mathbbm{1}$ with $s = \frac{1}{9}(5+2\sqrt{3}) < 1$. This implies that $\ket{\psi_{\max}}$ is non-Gaussian, since it is a pure state. We can assume w.l.o.g. that the Gaussian state achieving the maximum Gaussian approximation ratio is pure as discussed in Section \ref{sec:opt-Gauss}. Let us express any $4$-mode state as 
\begin{equation}
    \ket{\psi} = \alpha\ket{\psi_{\max}} + \sqrt{1-|\alpha|^2}\ket{\psi_{\perp}},
\label{eq:generalstate}
\end{equation}
with $\alpha \in \mathbb{C}$ and where $\ket{\psi_{\perp}}$ is any state s.t. $\braket{\psi_{\max}}{\psi_{\perp}} = 0$. Then, 
\begin{equation}
    \bra{\psi}H\ket{\psi} = |\alpha|^2\lambda_{\max}(H) + (1-|\alpha|^2)\bra{\psi_{\perp}}H\ket{\psi_{\perp}} \leq |\alpha|^2\lambda_{\max}(H) + (1-|\alpha|^2)(\lambda_{\max}(H)-\Delta). 
\end{equation}
Naturally, there is a value for $|\alpha|$ above which all states $\ket{\psi}$ in Eq.~\eqref{eq:generalstate} are non-Gaussian, hence $|\alpha|$ should be below this value to ensure Gaussianity, thus upperbounding $\bra{\psi}H\ket{\psi}$ achieved by any pure fermionic Gaussian state $\ket{\psi}$. To find such $|\alpha|$, we evaluate the entries of the covariance matrix of $\ket{\psi}$ in Eq.~\eqref{eq:generalstate}. Let $\{\tilde{c}_{j}\}_{j=1}^{8}$ be the Majorana basis in which $\Gamma_{\psi_{\max}}$ is ($2\times 2$) block-diagonal, so that in that basis $\Gamma_{\psi_{\max}} = \bigoplus_{i=1}^{4}\Big( \:^{\:\:\:0}_{-\lambda_{i}}\:\:^{\lambda_{i}}_{\:0} \Big)$ with $\lambda_{i}\in [-1,+1]$. For the covariance matrix $\Gamma_{\psi}$ of $\ket{\psi}$ in Eq.~ \eqref{eq:generalstate}, we then have for $j\neq k$
\begin{align}
    \big(\Gamma_{\psi}\big)_{j,k} &\:=  \bra{\psi}i\tilde{c}_{j}\tilde{c}_{k}\ket{\psi} \nonumber \\ &\hspace{-1.2cm} = |\alpha|^2\bra{\psi_{\max}}i\tilde{c}_{j}\tilde{c}_{k}\ket{\psi_{\max}} + 2\sqrt{1-|\alpha|^2}\:\text{Re}\big( \alpha \bra{\psi_{\max}}i\tilde{c}_{j}\tilde{c}_{k}\ket{\psi_{\perp}} \big) + (1-|\alpha|^2)\bra{\psi_{\perp}}i\tilde{c}_{j}\tilde{c}_{k}\ket{\psi_{\perp}} \nonumber \\ &\hspace{-1.2cm} = \begin{cases}
        \pm |\alpha|^{2}\lambda_i + (1-|\alpha|^2)\bra{\psi_{\perp}}i\tilde{c}_{j}\tilde{c}_{k}\ket{\psi_{\perp}}, & \text{for }(j,k) = (2i-1,2i) \text{ or } (2i,2-i), \\ 
        2\sqrt{1-|\alpha|^2}\:\text{Re}\big( \alpha \bra{\psi_{\max}}i\tilde{c}_{j}\tilde{c}_{k}\ket{\psi_{\perp}} \big) + (1-|\alpha|^2)\bra{\psi_{\perp}}i\tilde{c}_{j}\tilde{c}_{k}\ket{\psi_{\perp}}, & \text{elsewhere},
    \end{cases}
\end{align}
where we have used that $\ket{\psi_{\max}}$ is an eigenstate of $i\tilde{c}_{2i-1}\tilde{c}_{2i}$ for $i=1,2,3,4$. For any pure fermionic Gaussian state with covariance matrix $\Gamma$, we have that $\Gamma \Gamma^{T} = \mathbbm{1}$ (see Section \ref{sec:preliminaries}), so that $\sum_{k}\Gamma_{j,k}^{2} = 1$ $\forall j$. Using that $\bra{\psi_{\perp}}i\tilde{c}_{j}\tilde{c}_{k}\ket{\psi_{\perp}} \in [-1,+1]$ and $\text{Re}\big( \alpha \bra{\psi_{\max}}i\tilde{c}_{j}\tilde{c}_{k}\ket{\psi_{\perp}} \big) \in [-|\alpha|,+|\alpha|]$, we find 
\begin{equation}
    \sum_{k}\Gamma_{j,k}^{2} \leq \big(|\alpha|^2\sqrt{s}+(1-|\alpha|^2))^2 + 6\big(2|\alpha|\sqrt{1-|\alpha|^2} + (1-|\alpha|^2) \big)^2 = g_{s}(|\alpha|),
\end{equation}
where we have used that $|\lambda_{i}| = \sqrt{s}$ for any $i = 1,2,3,4$. We have $g_{s = \frac{1}{9}(5+2\sqrt{3})}(|\alpha_*|) = 1$ at $|\alpha_*| \approx 0.998818$, so that $\ket{\psi}$ in Eq.~\eqref{eq:generalstate} is non-Gaussian for $|\alpha| > |\alpha_*|$ where $g_{s = \frac{1}{9}(5+2\sqrt{3})}(|\alpha|)< 1$. Therefore, the approximation ratio achievable by Gaussian states is upper bounded by 
\begin{equation}
    \max_{0\leq |\alpha|\leq |\alpha_*|} \big[|\alpha|^2\lambda_{\max}(H) + (1-|\alpha|^2)(\lambda_{\max}(H)-\Delta)\big]/\lambda_{\max}(H) < 0.99904,
\end{equation}
where we have used $\lambda_{\max}(H) = \frac{1}{4}(3 + 2\sqrt{3})$ and $\Delta = \frac{1}{2}(1+\sqrt{3}-\sqrt{2})$. 
\end{proof}

We note that this type of argument does not work to bound the Gaussian approximation ratio for a system of growing size $n$, as the gap scaling is small relative to the scaling of the maximum energy. However, the antiferromagnetic Heisenberg model in 1D, i.e. \textsc{Quantum Max Cut} on a line, is a well studied model, solved via the Bethe ansatz, and we expect that when translated to fermions, its ground state is never Gaussian. Note that product states do not generally translate to fermionic Gaussian states via the Jordan-Wigner transformation, nor vice versa. 

We note that Proposition \ref{prop:gaussianupperbound} also gives a Gaussian upper bound for traceless \textsc{Fermionic Max Cut}, i.e., Problem \ref{prob:tracefulfermionicmaxcut} made traceless. 

Numerically, we find that there exists a fermionic Gaussian state that achieves approximation ratio \textit{at least} 0.9788 for the (traceless) $H$ in Eq.~\eqref{eq:Hgaussianupperbound}. 
How does this compare to a product state approximation ratio for this 1D \textsc{Quantum Max Cut} model?
Making $H$ positive semi-definite like in \textsc{Fermionic Max Cut}, this numerically obtained Gaussian state achieves the ratio $\approx 0.9855$. Since product states achieve at most approximation ratio $\frac{1}{2}$ on a single \textsc{Quantum Max Cut} edge $H_{i,i+1}$ and $\lambda_{\max}(H_{i,i+1}) = 1$, they achieve ratio at most $\frac{3/2}{\lambda_{\max}(H)}\approx 0.6340$ for $H$ the unit weight (positive semi-definite) \textsc{Quantum Max Cut} Hamiltonian on a $n=4$ line.  Thus, the relatively high approximation ratio achieved by a fermionic Gaussian state for \textsc{Quantum Max Cut} on a $n=4$ line suggests that fermionic Gaussian states might be used in approximations for \textsc{Quantum Max Cut} more generally (beyond 1D systems). 

\section{Example system with $\lvert\lambda_{\min}(H_{\rm class})\rvert/\lambda_{\max}(H_{\rm class}) = n$ for a traceless ${H}_{\rm class}$}
\label{sec:minvsmaxscalingclassicalcostfunction}
Consider the traceless classical Hamiltonian $H_{\rm class}$ in Eq.~\eqref{eq:classH} where $G_{\rm class}=((\mu,V),(w,E))$ is the complete graph with $w_{j,k} = w$ $\forall (j,k)\in E$ and $\mu_{j} = \mu$ $\forall j \in V$ and denote $n_j=x_j\in \{0,1\}$, a binary vector ${\bf x}$ of length $n$ to be optimized. Let $N = \sum_{j\in V}x_{j}$. For such instance, we have
\begin{align}
    H_{\rm class} = \frac{1}{2}w\sum_{j\neq k}\big(1/4-x_{j}x_{k}\big) + \mu \sum_{j \in V}\big(x_{j}-1/2\big) = \frac{1}{8} wn(n-1)-\frac{1}{2}w(N^2-N)+\mu(N-n/2).
\end{align}
This cost function is a concave function of $N$, whose maximum is achieved at $N = \frac{1}{2}+\mu/w$ and whose minimum is achieved at either $N=0$ or $N=n$. The associated values of the cost function are 
\begin{align}
    \lambda_{\max}(H_{\rm class}) =&\: \frac{1}{8}w\big(1+n(n-1)\big)+\frac{\mu^2}{2w} - \frac{\mu}{2}(n-1), \nonumber \\
    \lambda_{\min}(H_{\rm class}) =&\: \min\big\{ \frac{1}{8} wn(n-1)-\frac{\mu n}{2},-\frac{3}{8} wn(n-1)+\frac{\mu n}{2} \big\}.
\end{align}
Setting $\mu = wn/2$, this reduces to $\lambda_{\max}(H_{\rm class}) = \frac{1}{8} w(n+1)$ and $\lambda_{\min}(H_{\rm class}) = -\frac{1}{8} wn(n+1)$, so that $\lvert \lambda_{\min}(H_{\rm class}) \rvert / \lambda_{\max}(H_{\rm class}) = n$.


In contrast, for any \textit{sparse} traceless interactions $H_{\rm class}$ (with constant coefficients $w_{j,k},\mu_j$), we have $\lambda_{\max}(H_{\rm class}) = \Theta(n)$ ---and therefore $\lvert\lambda_{\min}(H_{\rm class})\rvert = \Theta(n)$ since $-H_{\rm class}$ is also sparse--- so that the ratio $\lvert\lambda_{\min}(H_{\rm class})\rvert/\lambda_{\max}(H_{\rm class})$ is a constant.

\section{SDP relaxation and rounding?}
\label{sec:SDPrelaxandround}
Another potential direction of further research is the following. Similar to the approaches used to optimize \textsc{Quantum Max Cut}, one can define an SDP hierarchy, see e.g. \cite{HO, hastings2024SOS}, optimizing $H$ over correlation matrices ---expressing correlations between weight-$k$ Majorana monomials--- of poly$(n)$ size for $k=O(1)$. These SDP's are relaxations of the eigenvalue problem s.t. $\text{SDP}_{k=1}\geq \text{SDP}_{k=2}\geq \ldots \geq \text{SDP}_{k=n} = \lambda_{\max}(H)$, with $\text{SDP}_{k}$ denoting the optimum of the SDP at level $k$. The feasible solutions of such SDP's do not correspond to valid density matrices in general, let alone to fermionic Gaussian states. If, however, we could round the optimum $M^{(k)}$ of $\text{SDP}_{k}$ to a fermionic Gaussian state $\rho_{\rm Gauss}$ s.t. $\tr(\rho_{\rm Gauss} H) = r\:\text{SDP}_{k} \geq r\:\lambda_{\max}(H)$ (with $0\leq r\leq 1$), then we have an $r$ approximation algorithm for Hamiltonian $H$. For traceless Hamiltonians, we know such a rounding approach does not exist for constant $r$ in general \cite{arora2005non}. Therefore, any such rounding scheme should leverage specific properties of the Hamiltonian at hand, such as it being positive semi-definite.  
For completeness, let us briefly discuss the semi-definite program that computes $\text{SDP}_{k}$. 

Let
\begin{align}
    C_I := i^{\binom{k}{2}}\: c_{i_1}\ldots c_{i_k},I=\{i_1 < i_2 < \ldots < i_k\},
\end{align} 
be a weight-$k$ Majorana monomial labeled by the ordered subset $I$, with $C_I^{\dagger}=C_I$, $C_{I}^{2} = \mathbbm{1}$ and $C_{I}C_{J} = (-1)^{|I|\:|J| - |I\cap I'|}C_{J}C_{I}$. 
The spectral norm of operators $C_I$ or $C_I C_J$ is at most $1$.

The collection of monomials $C_I$ of weight at least 1 and at most $k$ is denoted by $\mathcal{C}_{k}\subseteq \mathcal{C}_{2n}$. Clearly, the number of monomials in $\mathcal{C}_{k}$ is $|\mathcal{C}_{k}| = \sum_{m=1}^{k}\binom{2n}{m} = \Theta(n^{k})$. For any Hermitian matrix $\sigma$ with $\tr \sigma=1$ (not necessarily a density matrix), we define the weight-$k$ correlation matrix 
\begin{equation}
   M_{I,J}^{(k)}(\sigma) := \tr\big(C_{I}^{\dagger}C_{J} \sigma  \big) \in \mathbb{C},\: \forall I \; M^{(k)}_{I,I}=1,\: |M^{(k)}_{I,J}|\leq 1,
   \label{eq:M}
\end{equation}
with $C_{I},C_{J}\in \mathcal{C}_{k}$. $M^{(k)}\succeq 0$ by construction, since
$\boldsymbol{\alpha}^{\dagger}M^{(k)}\boldsymbol{\alpha} = \sum_{I,J} \alpha^*_I \alpha_{J} M_{I,J}=\tr (E^{\dagger}E\,\sigma )\geq 0$, with $E=\sum_I \alpha_I C_I$, for any complex vector $\boldsymbol{\alpha}$. For $k=n$, $M^{(k)}$ can be associated with a valid density matrix $\sigma$. 

The constraints on $M$ are those of the feasible set of an SDP, i.e. $M^{(k)}\succeq 0$, obeying some linear equality constraints related to anti-commutation, or product rules of the $C_I$ operators, as well the linear inequality constraints given in Eq.~\eqref{eq:M}. This prescribed range of the entries is bounded appropriately by the constraints $M^{(k)}_{I,I} = 1$ and $M^{(k)}\succeq 0$ (since all principal minors of a positive semi-definite matrix are non-negative). 
Let us refer to this feasible set of positive semi-definite matrices $M^{(k)}$ as $\mathcal{L}_k$.

Given a quartic fermionic Hamiltonian $H=\sum_{I,J \colon |I|,|J|\leq 2} h_{J,I} C_I^{\dagger} C_J$ (which has non-zero trace in general since $C_{I}^{2} = \mathbbm{1}$), with $h_{J,I}\in \mathbb{C}$ and $h$ s.t. $H$ contains only quadratic and quartic terms in the $\{c_i\}$, one can write, for any fermionic density matrix $\rho \succeq 0$
\begin{align}
\tr \rho H=\sum_{I,J} h_{J,I}M^{(2)}_{I,J}(\rho)=\tr (h M^{(2)}(\rho)).    
\end{align}
Hence, optimization over fermionic density matrices $\rho$ can be relaxed to optimization over weight-$k$ correlation matrices, i.e. one defines the hierarchy $\text{SDP}_{k} = \sup_{M^{(k)}}  \tr(h M^{(2)}) \text{ s.t.}\; M^{(k)} \in \mathcal{L}_k$ whose solution may correspond to $\sigma \not\succeq 0$ (sometimes called a pseudo-density matrix).

An approach used in \textsc{Quantum Max Cut} \cite{gharibian_parekh} to obtain an (almost optimal) product state solution is to employ a randomized rounding scheme similar to GW rounding \cite{GW}. An analogous qubit SDP hierarchy is defined and the optimum weight-$1$ Pauli correlation matrix is Cholesky decomposed into $3n$-dimensional vectors. Then, these vectors are rounded to $n$ $3$-dimensional (normalized) vectors ---which directly relate to the expectation with respect to a product state solution that achieves an approximation ratio close to the optimal one over product states. One might wonder whether a similar approach could be applied to \textsc{Fermionic Max Cut} to obtain good Gaussian approximations. The fact that one rounds to Gaussian states instead of e.g. product states makes it essentially different. Namely, given a Cholesky decomposition $\{\mathbf{b}_{I}\}$ of an optimum $M^{(k)}$ (s.t. $M^{(k)}_{I,J} = \mathbf{b}_{I}^{\dagger}\mathbf{b}_{J}$) of the fermionic SDP hierarchy, one needs to randomly project these vectors $\mathbf{b}_{I}$ to a valid Gaussian covariance matrix. Naive attempts at such rounding procedures do not necessarily yield valid Gaussian covariance matrices. 


\end{document}